\documentclass[final,3p,times,twocolumn]{elsarticle}
\usepackage[english]{babel}
\usepackage{amsmath,amsthm,amssymb}
\usepackage{tikz}
\usepackage{tikz-cd}
\usepackage{enumitem}
\usepackage{mathrsfs}
\usepackage{dirtytalk}
\usepackage{multirow}
\usepackage{etoolbox}
\usepackage{ragged2e}
\usepackage{empheq}
\usepackage{subcaption}

\newtheoremstyle{theorem}
                {4pt}
                {4pt}
                {}
                {}
                {\bfseries}
                {\\}
                {10pt}
                {}

\theoremstyle{theorem}
\newtheorem{definition}{Definition}
\newtheorem{assumption}{Assumption}
\newtheorem{theorem}{Theorem}

\newtheorem{lemma}{Lemma}

\newtheorem{remark}{Remark}
\renewcommand{\qedsymbol}{$\square$}

\journal{Systems and Control Letters}
\begin{document}
	\justifying 
	\begin{frontmatter}
		\title{Towards Polynomial Immersion of Port-Hamiltonian Systems}
		\author{Mohammad Itani$\hspace{0.2mm}^a$} 
        \author{Manuel Schaller$\hspace{0.4mm}^b$}
        \author{Karl Worthmann$\hspace{0.2mm}^c$}
        \author{Timm Faulwasser$\hspace{0.4mm}^a$}
        \affiliation{organization={Hamburg University of Technology, Institute of Control Systems},city={Hamburg},country={Germany}}
        \affiliation{organization={Chemnitz University of Technology},city={Chemnitz},country={Germany}}
		\affiliation{organization={Ilmenau University of Technology},city={Ilmenau},country={Germany}}
        %----------------------------------------------------
		%% Abstract
		\begin{abstract}
            Port-Hamiltonian (pH) systems offer a highly structured and energy-based modular framework for control systems.
			Many pH systems exhibit non-polynomial non-linearities. We consider the problem of immersing such systems into a higher-dimensional polynomial representation.            
            We prove that, along system trajectories, important features of the non-polynomial pH system are preserved such as the internal interconnection geometry, the energy balance relation with passivity supply rate, as well as energy dissipation. We illustrate how the lifted system enables the design of stabilizing feedback laws by combining sum-of-squares optimization with concepts from passivity-based control.  We draw upon several examples to illustrate our findings. 
		\end{abstract}
	%----------------------------------------------------	
		%% Keywords
	\end{frontmatter}
    %----------------------------------------------------
    \begingroup
    \renewcommand{\thefootnote}{}
    \footnotetext{The authors acknowledge support by the Deutsche Forschungsgemeinschaft -- DFG  (project number 519323897).}
    \endgroup
    %----------------------------------------------------
	\section{Introduction} \label{section:intro}
    While control design and stability analysis are challenging for general nonlinear control systems, constructive methods are available in the case of polynomial nonlinearities, see, e.g., \cite{parrilo2000structured,papachristodoulou2002,prajna2004}. 
    For systems whose dynamics are non-polynomial, recasting techniques such as lifting and immersion can be applied to obtain a polynomial representation. Lifting refers to the extension of the original system state with auxiliary coordinates such that the lifted dynamics are described by polynomials \cite{kramer2025discovering, gu2011qlmor, Papachristodoulou2005}. 
    Lifting always leads to additional nonlinear algebraic constraints that ensure correspondence to the original system. 
    When only the input-output behavior is of interest, a nonlinear system may also be immersed into another system usually of higher dimension \cite{fliess1983finiteness,ohtsuka2005model}. In this case, if the initial condition is consistently initialized, both systems admit the same input-output behavior without any additional algebraic constraints. 

    In this paper, we propose a lifting approach for port-Hamiltonian (pH) systems~\cite{duindam2009modeling,van2014port,maschke1992port}. Many physical systems can be modeled as nonlinear port-Hamiltonian systems, which reflect physical properties such as the interconnection geometry as well as energy dissipation and energy storage. For more recent works on the topic, see, e.g., \cite{philipp2024optimal,esterhuizen2024existence}.
    
    When applied to port-Hamiltonian systems, recasting techniques may destroy or obscure their favorable structural properties such as dissipativity or the conjugated input-output structure. Hence, we investigate a novel technique that combines conceptual ideas from lifting and immersion (hence called lifted immersion) which preserves the pH structure with the same rational Hamiltonian as the original system as well as the internal interconnection structure (modulo a trivial lifting to a higher dimension). The \textit{target} system obtained by the proposed method admits a polynomial structure, i.e., all functions defining the dynamics and outputs of the system are polynomials in the lifted state. While the external interconnection structure of the resulting system contains additional non-trivial terms, the new output is identical to the original output and admits the usual conjugated port structure encountered in pH systems. 
    We show that the dissipation properties of the original system hold along the new system's trajectories if the initial condition is consistently initialized. Thus, passivity of the resulting polynomial system with respect to the original Hamiltonian follows on the range of the lifted immersion which defines an invariant embedded submanifold in the lifted state space.
    
    Finally, we illustrate via an example the potential of polynomial pH system structures for control design combining interconnection and damping assignment \cite{ortega2002interconnection} with sum-of-squares optimization \cite{parrilo2000structured, papachristodoulou2002,Papachristodoulou2005}.

    The remainder of the paper is organized as follows: 
    Section \ref{sec:preliminaries} introduces the problem statement and provides an overview on immersion notions. 
    Section \ref{sec:main_results} presents our approach of constructing a structure-preserving polynomial immersion for pH systems with two tutorial examples. In Section \ref{sec:ida_sos}, we illustrate, via an example, how polynomial immersions enable the design of stabilizing controllers via sum-of-squares techniques. Finally, conclusions and outlook are presented in Section \ref{sec:conclusion}.
%-----------------------------------------------------------------	

	\paragraph{Notation}
    A subscript for the sets $\mathbb{N}$ and $\mathbb{Z}$ denotes an upper bound, e.g., $\mathbb{N}_m := \{1,\ldots,m\}$. The sets $\mathbb{R}^+$ and $\mathbb{N}^0$, respectively, denote the set of nonnegative real numbers and the set $\mathbb{N} \cup \{0\}$. $\mathbb{R}[x], \, x \in \mathbb{R}^n,$ denotes the ring of polynomials (over $\mathbb{R}$) with variables $x_1,\ldots,x_n$ and is $\mathbb{R}(x)$ the corresponding field of rational functions. For a smooth function $\alpha: \mathbb{R}^n \rightarrow \mathbb{R}^r$, $\mathbb{R}[\alpha(x)]$ ($\mathbb{R}(\alpha(x))$) denotes the ring of polynomials whose generators are $\alpha_1(x), \ldots, \alpha_r(x)$ (the corresponding rational field). $\mathbb{R}[\cdot]^{n}$, respectively, $\mathbb{R}[\cdot]^{n\times m}$  denote $n$-dimensional vectors and $n \times m$-dimensional matrices whose elements are polynomials in $(\cdot)$. The same holds for $\mathbb{R}(\cdot)^{n}$ and $\mathbb{R}(\cdot)^{n\times m}$. Let $H: \mathbb{R}^n \rightarrow \mathbb{R}$ be a differentiable function, then $\nabla H: \mathbb{R}^n \rightarrow \mathbb{R}^n, \, x \mapsto \left[\frac{\partial H(x)}{\partial{x}_1} \; \ldots \, \frac{\partial H(x)}{\partial{x}_n}\right]^\top$, is its gradient. For a differentiable vector-valued map, $\Psi: \mathbb{R}^n \rightarrow \mathbb{R}^{\tilde{n}}$, its Jacobian is the $\tilde{n}\times n$ matrix 
    \(
        \operatorname{D}\Psi(x) := 
        \operatorname{col}(\nabla \Psi_1^\top(x),\ldots, \nabla \Psi_{\tilde{n}}^\top(x))
    \).
%----------------------------------------------------------------------
	\section{Preliminaries and Problem Statement}\label{sec:preliminaries}
    In this section, we present the problem formulation, then we introduce the concept of system immersion.
    
    \subsection{Problem Statement}
    \label{subsec:problem_statement}
	We consider nonlinear port-Hamiltonian systems
    \vspace{-0.2cm}
    \begin{subequations} \label{eqn:pH_system} 
	\begin{align} 
				\dot{x} &= \left( J(x)- R(x) \right) \nabla H(x) + \sum\limits_{i=1}^{m} g_i(x) u_i, \; x(t_0) = x_0 
                \label{eqn:ODE_pH}
                \\
				y &= \operatorname{col}(g_1^\top(x),\ldots,g_m^\top(x))\nabla H(x) \label{eqn:ODE_output}, 
	\end{align} 
        \end{subequations}
	where $x \in \mathcal{D}, \; \mathcal{D} \subseteq \mathbb{R}^n$ is an open set, and $u, y \in \mathbb{R}^m$. 
    The Hamiltonian $H: \mathbb{R}^{n} \rightarrow\mathbb{R}^{+}$ captures the total energy of the system, and the matrices $J(x), \, R(x) \in \mathbb{R}^{n \times n}$, respectively, model its internal interconnection and dissipation structure. For all $x \in \mathcal{D}$, they satisfy 
    \[
    J^\top(x) = -J(x) \quad\text{ and }\quad R(x)=R^\top(x) \succeq0.
    \]
    Moreover, $g_i:\mathcal{D} \rightarrow \mathbb{R}^n, \, i \in \mathbb{N}_{m}$, are the control vector fields. We consider piecewise continuous inputs functions denoted by $\mathcal{U}= \mathcal{PC}([t_0,t_f],\mathbb{R}^m)$. We assume that all functions on the right-hand side of~\eqref{eqn:pH_system} are real-analytic on $\mathcal{D}$. A function $f:\mathcal{D}\rightarrow \mathbb{R}$ is called real-analytic on some open set $\mathcal{D} \subseteq \mathbb{R}^n$ if its Taylor series at $z_0$ (for every $z_0 \in \mathcal{D}$) converges to $f(z)$ in a neighborhood of $z_0$. If $f$ is vector-valued, then it is real-analytic if each of its components is real-analytic. It is worth noting that real-analytic functions are smooth functions; see \cite{krantz2002primer} for more details on the topic. Henceforth, we refer to real-analytic functions as analytic functions.
    
	In this work, we investigate the following questions: 
    \begin{itemize}
        \item[i)] \emph{First, is it possible to find a polynomial representation of \eqref{eqn:pH_system} without sacrificing its port-Hamiltonian structure?}
        \item[ii)] \emph{Second, how does one leverage the obtained polynomial structure for control design?}
    \end{itemize}
    
    \subsection{Immersions}
    \label{subsec:immersion}
    We first recall the concept of immersion, which will be the central ingredient for our lifting approach. 
    Consider nonlinear control-affine systems of the form
	\begin{subequations} \label{eqn:affine_control_sys}  
    \vspace{-0.2cm}
    \begin{align}
            \dot{x} &= g_0(x) + \sum\limits_{i=1}^{m} g_i(x) u_i, \quad x(t_0) = x_0 
            \label{eqn:affine_ODE}
            \\
			y &= h(x), 
        \end{align}
	\end{subequations}
    where $x \in \mathcal{D}$ (an open set in $\mathbb{R}^n$) and all functions are analytic on $\mathcal{D}$.
    
    Note that the pH system \eqref{eqn:pH_system} is a special case of \eqref{eqn:affine_control_sys}.
    
    Due to the analyticity assumed above, for any $x_0 \in \mathcal{D}$ and $u(\cdot) \in \mathcal{U}$, there always exists (locally) a unique maximal solution for \eqref{eqn:affine_ODE} denoted by $x(\cdot, x_0,u) = x(t, x_0,u), \, t \in [t_0,t_m]$, or briefly as $x(\cdot)$.
    
\begin{definition}[Immersibility \cite{fliess1983finiteness,ohtsuka2005model}]   
\label{def:immersion}
      Consider a system of the form 
    \vspace{-0.2cm}
    \begin{subequations}
        \label{eqn:imm}
        \begin{align}
        \dot{\tilde{x}} &= \tilde{g}_0(\tilde{x}) + \sum\limits_{i=1}^{m} \tilde{g}_i(\tilde{x}) u_i, 
        \quad \tilde{x}(t_0) = \tilde{x}_0 \label{eqn:poly_imm_dyn} \\
        \tilde{y} &= \tilde{h}(\tilde{x}), \label{eqn:poly_imm_out}
        \end{align}
    \end{subequations}
    where $\widetilde{\mathcal{D}}$ is an open set of $\mathbb{R}^{\tilde{n}}$, $\tilde{x} \in \widetilde{\mathcal{D}}$, and the functions $\tilde{g}_0, \ldots,\tilde{g}_m,\tilde{h}$ are analytic on $\widetilde{\mathcal{D}}$.

    System \eqref{eqn:affine_control_sys} is said to be immersible into \eqref{eqn:imm}, if there exists an analytic map $\Psi: \mathcal{D} \subseteq \mathbb{R}^{n} \to \mathbb{R}^{\tilde{n}}$, such that $\Psi(\mathcal{D}) \subseteq \widetilde{\mathcal{D}}$ and, for every $x_0 \in \mathcal{D}$, we have
    \[
        h \circ x(\cdot,x_0,u) = \tilde{h} \circ \tilde{x}(\cdot,\Psi(x_0),u).
    \] 
    Moreover, $\Psi$ is called an \textit{immersion of \eqref{eqn:affine_control_sys} into \eqref{eqn:imm} on $\mathcal{D}$}. \qed
\end{definition}  

\begin{figure*}[htbp]
        \centering
        \includegraphics[scale=0.75]{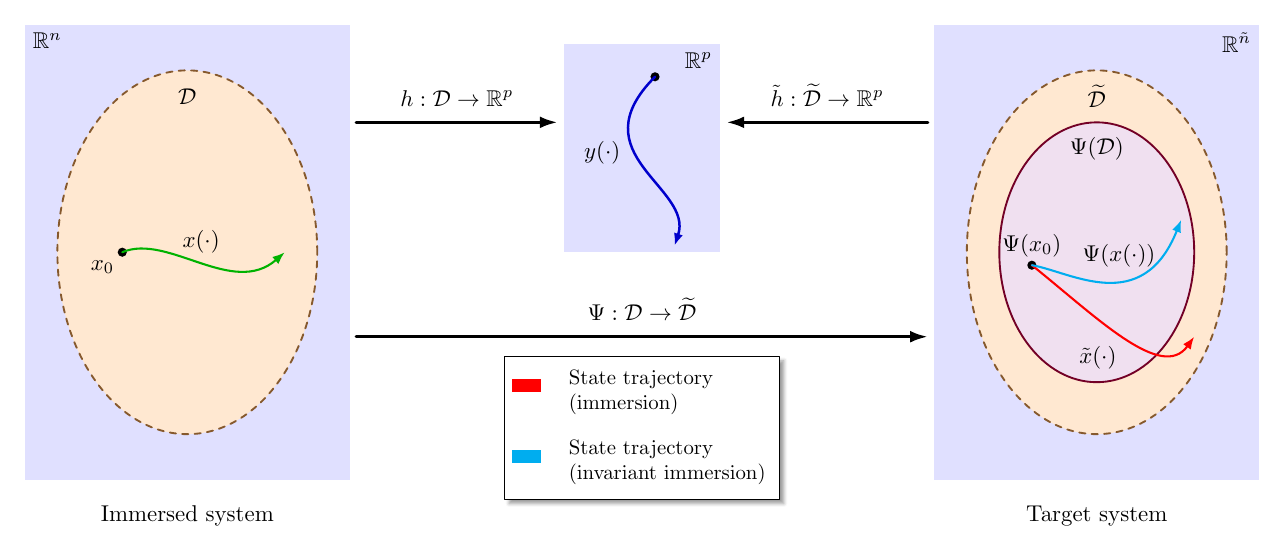}
        \caption{Illustration of the difference between an immersion and an invariant immersion.}
        \label{fig:comparison_immersion_inv}
\end{figure*}

    Immersions as per Definition~\ref{def:immersion} map the initial conditions of system \eqref{eqn:affine_control_sys} to initial conditions of \eqref{eqn:imm} such that the output trajectories of both systems are identical. Notice that the respective vector fields, as well as the output maps and state trajectories of \eqref{eqn:affine_control_sys} and \eqref{eqn:imm}, are not necessarily directly related via the immersion $\Psi$. This is in contrast to invariant immersions defined next.

    \begin{definition}[Invariant immersibility \cite{ohtsuka2005model}]
        \label{def:inv_imm}
        Consider a system \eqref{eqn:imm} and an analytic immersion $\Psi$ on $\mathcal{D}$ satisfying Definition \ref{def:immersion}. The immersion $\Psi$ is called \textit{invariant} if (for all $x \in \mathcal{D}$) the vector fields of \eqref{eqn:imm} satisfy 
        \[
            \tilde{g}_j(\tilde{x}) = \operatorname{D}\Psi(x)g_j(x), \, \forall j \in \mathbb{N}_m^0,
        \]
        and the output of \eqref{eqn:imm} satisfies $\tilde{h}(\tilde{x}) = h(x)$. \qed
    \end{definition}
    \noindent It follows from Definition \ref{def:inv_imm} that
    \begin{subequations}
        \label{eqn:inv_imm_dyn}
        \begin{align}
        \dot{\tilde{x}}(t) &= \operatorname{D}\Psi(x)\left(g_0(x(t)) + \sum\limits_{i=1}^{m} g_i(x(t)) u_i\right) \\
            & = L_{g_0} \Psi(x) + \sum\limits_{i=1}^{m} L_{g_i} \Psi(x) u_i \\
            &= \frac{d}{dt} \Psi(x).
        \end{align}
    \end{subequations}
    Thus, for a state trajectory $x(\cdot, x_0,u)$ of \eqref{eqn:affine_control_sys}, $\Psi(x(\cdot, x_0,u))$ is a corresponding state trajectory of \eqref{eqn:inv_imm_dyn}. More precisely, if $\tilde{x}(\cdot,\Psi(x_0),u)$ denotes the solution of \eqref{eqn:inv_imm_dyn} with initial condition $\tilde{x}_0 = \Psi(x_0)$ and input $u \in \mathcal{U}$, the following point-wise relation holds   
    \[
        \tilde{x}(t,\Psi(x_0),u) = \Psi(x(t, x_0,u))
    \]
    for all $t$ for which the solution  $x(\cdot)$ exists. 
    This implies that $\Psi(\mathcal{D})$ is an invariant set for \eqref{eqn:inv_imm_dyn} for all $\tilde{x}_0 \in \Psi(\mathcal{D})$ and $u \in \mathcal{U}$ \cite{ohtsuka2005model}. Notice that this property does not necessarily hold for a generic immersion (Definition \ref{def:immersion}) since $\tilde{x}(t) = \Psi(x(t))$ might not necessarily hold for $t>t_0$. Figure \ref{fig:comparison_immersion_inv} illustrates this crucial distinction. It is also worth noting that, for invariant immersions, the corresponding vector fields of the original system \eqref{eqn:affine_control_sys} and the new system \eqref{eqn:imm} are related by taking Lie derivatives of the immersion in the direction of the original system vector fields. Moreover, the output $\tilde{y}$ of \eqref{eqn:imm} is given by the pullback of $\tilde{h}$ by the immersion $\Psi$, i.e., 
    \[
    \tilde{y} = \tilde{h} \circ \Psi(x).
    \]
    %----------------------------------------------
    \begin{remark}[Manifold immersions] 
    Definitions \ref{def:immersion} and \ref{def:inv_imm} concern immersions of dynamic systems. On the other hand, an immersion between smooth manifolds $M$ and $N$ is a smooth map $\varphi \in \mathcal{C}^\infty(M,N)$ such that \cite{tu2011manifolds}
    \begin{equation}
    \label{eqn:rank_cond}
    \forall x \in M \; : \; 
    \operatorname{rank} \dfrac{\partial{\varphi}}{\partial{x}} = \operatorname{dim}M.
    \end{equation}
    Nevertheless, immersions per Definitions \ref{def:immersion} and \ref{def:inv_imm} are also manifold immersions -- between the state spaces of the immersed system and the target system -- if an observability rank condition of \eqref{eqn:affine_control_sys} holds. We refer to \cite{ohtsuka2005model} for details. \qed
    \end{remark}
    %-----------------------------------------------
    Henceforth, we consider a special case of invariant immersions, namely, \emph{differential algebraic} immersions \cite{ohtsuka2005model,ohtsuka2002differentially}. To this end, we recall the notion of a differential field; for further details see \cite{ohtsuka2005model,ohtsuka2002differentially,ritt1950differential}. 
    
    Consider a vector function $\alpha:\mathcal{D} \rightarrow \mathbb{R}^\ell, \; x \mapsto \alpha(x)$, where $\mathcal{D} \subseteq \mathbb{R}^n$ is an open set. Assume $\alpha$ is analytic on $\mathcal{D}$ and denote the set of partial differential operators with respect to $x_1, \ldots, x_n$ by
    \[
        \nabla := \left\{ \partial / \partial{x}_1, \ldots, \partial / \partial{x}_n \right\}.
    \] 
    Then, by $\mathbb{R}\langle \alpha_1(x), \ldots, \alpha_\ell(x) \rangle_\nabla$, or equivalently $\mathbb{R}\langle \alpha(x) \rangle_\nabla$, we denote the (partial) differential field over $\mathbb{R}$ whose generators are $\alpha_1(x), \ldots, \alpha_\ell(x)$. That is, it is the field over $\mathbb{R}$ generated by $\alpha_i(x), \; i \in \mathbb{N}_\ell$, and equipped with the derivation\footnote{A derivation is a function notion that generalizes the standard differentiation notion in Euclidean spaces to algebras over rings.} operators in $\nabla$.  

    \begin{definition}[Differential algebraic immersion \cite{ohtsuka2005model,ohtsuka2002differentially}]
        \label{def:DA_immmersion}
        An invariant immersion $\Psi$ (Definition \ref{def:inv_imm}) is a \textit{differential algebraic (DA) immersion} if there exist a matrix $P$ and analytic functions $\xi_j$ (on $\widetilde{\mathcal{D}}$) such that 
        \[
        P(\tilde{x}) = \operatorname{D}\Psi(x) \in \mathbb{R}(\tilde{x})^{\tilde{n} \times n} \quad\text{ and }\quad \xi_j(\tilde{x})=g_j(x), \, j \in \mathbb{N}_m^0.
        \]
    \end{definition}
		 
    Notice that in Definition \ref{def:DA_immmersion} the Jacobian $\operatorname{D}\Psi(x)$ and the vector fields $g_j(x)$ can be written purely as function of the immersion $\tilde{x}:= \Psi(x)$. Moreover, the entries of the former are purely rational in $\tilde{x}$. Henceforth, we exclusively consider DA immersions. 

    We are primarily interested in obtaining \emph{polynomial} and \emph{rational} representations of systems. The  control-affine system \eqref{eqn:imm} is polynomial if all defining functions are polynomials of the state $x$, i.e.,
    \begin{equation} 
        \label{eqn:poly_sys}
        \tilde{g}_j(\tilde{x}) \in \mathbb{R}[\tilde{x}]^{\tilde{n}}, \, \forall j \in \mathbb{N}_m^0 \; \text{ and } \; \tilde{h}(\tilde{x}) \in \mathbb{R}[\tilde{x}]^p,
    \end{equation}
    where $\mathbb{R}[\cdot]$ is the ring of polynomials over $\mathbb{R}$ with indeterminates $(\cdot)$. 
    Similarly, \eqref{eqn:imm} is called a rational system if
    \begin{equation}
        \label{eqn:rat_sys}
        \tilde{g}_j(\tilde{x}) \in \mathbb{R}(\tilde{x})^{\tilde{n}}, \, \forall j \in \mathbb{N}_m^0 \; \text{ and } \; \tilde{h}(\tilde{x}) \in \mathbb{R}(\tilde{x})^p,
    \end{equation} 
    where $\mathbb{R}(\cdot)$ denotes the field of fractions over $\mathbb{R}$ with indeterminates $(\cdot)$. 

    A necessary and sufficient condition for the existence of a DA immersion of \eqref{eqn:affine_control_sys} into a rational system \eqref{eqn:rat_sys} is that \eqref{eqn:affine_control_sys} is composed solely of \emph{differential algebraic functions} \cite{ohtsuka2005model}. 

    \begin{definition}[Differential algebraic function 
    \label{def:differential_algebraic_fncs}\cite{ohtsuka2005model}]
        A scalar analytic function $\alpha: \mathcal{D} \rightarrow \mathbb{R}$, $x \mapsto \alpha(x)$, where $\mathcal{D} \subseteq \mathbb{R}^n$ is an open set, is called differential algebraic if for every $i \in \mathbb{N}_n$ it satisfies an \emph{algebraic differential equation} of the form
        \begin{equation}
            \label{eqn:ADE}
            p_i\left(x,\alpha(x),\dfrac{\partial{\alpha(x)}}{\partial{x}_i}, \ldots,\dfrac{\partial^{\hspace{0.2ex}\rho_i}\alpha(x)}{\partial{x}_i^{\hspace{0.3ex}\rho_i}}\right)=0, 
        \end{equation}
        where $p_i(\cdot) \in \mathbb{R}[\cdot] \setminus\{0\}, \; \rho_i \in \mathbb{Z}^+$. \qed
    \end{definition}
    %----------------------------------------------------------
    \begin{remark}[Properties of diff. algebraic functions]
        \label{rmk:DA_fncs}
        A single-variable analytic function $\zeta(z)$ is called differential algebraic if it satisfies an \emph{algebraic differential equation} of the form $p(z, \zeta, \zeta',\ldots, \zeta^{(\rho)})=0$ where $p(\cdot) \in \mathbb{R}[\cdot] \setminus \{0\}$ and $\rho \in \mathbb{Z}^+$. In the multi-variable case, the function must satisfy a family of algebraic differential equations where each variable is treated once at a time while the others are held constant \cite{rubel1985differentially}. This fact is formalized in Definition \ref{def:differential_algebraic_fncs}. For vector functions, a function is said to be differential algebraic if each constituting scalar function is differential algebraic per Definition \ref{def:differential_algebraic_fncs}. Generally speaking, most functions that appear in the modeling of physical processes are differential algebraic. For example, polynomials, rational functions, trigonometric functions, exponential functions, and logarithmic functions are all differential algebraic functions. Analytic functions that are not differential algebraic are called transcendentally (or differentially) transcendental functions, e.g., Euler's Gamma function \cite{rubel1989survey}. Moreover, the set of differential algebraic functions is closed under addition, multiplication, composition, inversion, differentiation, and integration \cite{rubel1985differentially,ostrowski1920dirichletsche}. \qed
    \end{remark}
    %----------------------------------------------------------
    \begin{theorem}[Immersibility into rational systems \cite{ohtsuka2005model}] 
    \label{thm:rat_DA_imm} 
        Given a system \eqref{eqn:affine_control_sys} whose functions are all differential algebraic, i.e., 
            \begin{equation*}
                \label{eqn:DA_fncs_affine}
               g_{jk}(x), \, \forall j \in \mathbb{N}_m^0, \, k \in \mathbb{N}_n \; \text{ and } \; h_i(x), \; \forall i \in \mathbb{N}_p 
            \end{equation*}
        satisfy \eqref{eqn:ADE}. 
        Then there exists a finitely-generated rational field $\mathbb{R}(\psi_1(x),\ldots,\psi_{\tilde{n}}(x))$ with $\psi_1(x), \ldots, \psi_{\tilde{n}}(x)$ satisfying \eqref{eqn:ADE} that is equivalent to the differential field $\mathbb{R}\langle g_0(x), \ldots, g_m(x),h(x) \rangle_\nabla$. Moreover, the function \[\Psi(x) = [\psi_1(x),\ldots,\psi_{\tilde{n}}(x)]^\top\] defines a DA immersion (which is an invariant immersion) of \eqref{eqn:affine_control_sys} into a rational system \eqref{eqn:rat_sys}. \qed 
    \end{theorem} 
    %----------------------------------------------------------
    The next result shows that a rational system can always be immersed into a polynomial system.

    \begin{theorem}[\mbox{Immersibility into polynomial systems~\cite{ohtsuka2005model}}]
    \label{thm:poly_imm}
        Let $\Psi(x) = [\psi_1(x),\ldots,\psi_{\tilde{n}}(x)]^\top$ define an immersion of \eqref{eqn:affine_control_sys} into a rational system \eqref{eqn:rat_sys}. Then, $\Psi(x)$ can always be extended to define an immersion of \eqref{eqn:affine_control_sys} into a polynomial system \eqref{eqn:poly_sys}. \qed
    \end{theorem}
    For the sake of completeness, and since a detailed proof is not given in \cite{ohtsuka2005model}, we provide one here.
    \begin{proof}
         The functions defining \eqref{eqn:rat_sys} can be written in the form
        \begin{align*}
            \tilde{g}_{jk}(\tilde{x}) &= \dfrac{a_{jk}(\tilde{x})}{b_{jk}(\tilde{x})}, \; \forall j \in \mathbb{N}_m^0, \, k \in \mathbb{N}_{\tilde{n}} \\
            \tilde{h}_i(\tilde{x}) &= \dfrac{c_i(\tilde{x})}{d_i(\tilde{x})}, \; \forall i \in \mathbb{N}_p,
        \end{align*}
        where
        \[
            a_{jk}(\tilde{x}), \, c_i(\tilde{x}) \in \mathbb{R}[\tilde{x}] \text{ and } 
            b_{jk}(\tilde{x}), \, d_i(\tilde{x}) \in \mathbb{R}[\tilde{x}] \setminus \{0\}.
        \]
        It is worth noting that by definition of an immersion (Definition \ref{def:immersion}), there exists $\widetilde{\mathcal{D}} \supseteq \Psi(\mathcal{D})$ such that $b_{jk}(\tilde{x}), \, d_i(\tilde{x}) \neq 0, \, \forall \tilde{x} \in \widetilde{\mathcal{D}}$. Consider, for all $b_{jk}(\tilde{x}), d_i(\tilde{x}) \in \mathbb{R}[\tilde{x}] \setminus \mathbb{R}$, the set 
        \[
        \mathscr{S}_1 = \{ 1/b_{jk}(\tilde{x}), \, 1/d_i(\tilde{x}), \forall j \in \mathbb{N}_m^0, \, k \in \mathbb{N}_{\tilde{n}}, \, i \in \mathbb{N}_p\}.
        \]
        In practice, one removes any redundant elements in $\mathscr{S}_1$, i.e., by using the following subset instead
        \[
                \Big\{ \alpha \in \mathscr{S}_1 : 
                \forall s \in \mathscr{S}_1 \setminus \{\alpha\}, \, \alpha \neq \beta s^\nu,  \beta \in \mathbb{R}, \nu \in \mathbb{N}  
                \Big\}.
        \]
        Without loss of generality, we implicitly assume this step while sticking to the notation used in $\mathscr{S}_1$. 
        Define the map 
        \[  
            \widehat{\Psi}(\tilde{x}) :=
            \begin{bmatrix}
                 \tilde{x}^\top & \dfrac{1}{b_{jk}(\tilde{x})} & \dfrac{1}{d_{i}(\tilde{x})}
            \end{bmatrix}^\top = \hat{x}.
        \]
        Note that $\operatorname{dim}\hat{x} = \hat{n} = \tilde{n} + \kappa$ where $\kappa \leq \tilde{n}(m+1)+p$ denotes the number of functions/coordinates added besides $\tilde{x}$. It follows that $\tilde{g}_{jk}(\tilde{x}), \, \tilde{h}_i(\tilde{x}) \in \mathbb{R}[\hat{x}]$ or equivalently $\dot{\tilde{x}} \in \mathbb{R}[\hat{x}]^{\tilde{n}}, \, y \in \mathbb{R}[\hat{x}]^p$. 
        
        Thus, the dynamics of $\hat{x}$ read
        \[
            \dot{\hat{x}} =
            \begin{bmatrix}
                \dot{\tilde{x}}^\top & \dfrac{\mathrm{d}}{\mathrm{d}\hspace{0.05em}t}\left(\dfrac{1}{b_{jk}(\tilde{x})}\right) & \dfrac{\mathrm{d}}{\mathrm{d}\hspace{0.05em}t}\left(\dfrac{1}{d_{i}(\tilde{x})}\right)
            \end{bmatrix}^\top.
        \]
        The chain rule gives
        \[
        \dfrac{\mathrm{d}}{\mathrm{d}\hspace{0.05em}t}\left(\dfrac{1}{b_{jk}(\tilde{x})}\right) = \dfrac{\partial}{\partial{\tilde{x}}}\left(\dfrac{1}{b_{jk}(\tilde{x})}\right)\cdot \dot{\tilde{x}}.
        \]
        Notice that
        \[
            \dfrac{\partial{}}{\partial{\tilde{x}_k}} \left( \dfrac{1}{b_{jk}(\tilde{x})} \right) = -\dfrac{\partial{b}_{jk}(\tilde{x})/\partial{\tilde{x}}_k}{b_{jk}(\tilde{x})^2} \in \mathbb{R}[\hat{x}].
        \] 
        The same holds for $\dfrac{\mathrm{d}}{\mathrm{d}t} \left( \dfrac{1}{d_{i}(\tilde{x})} \right)$. Hence, it follows that 
        \begin{align*}
            \frac{\partial \widehat{\Psi}(\tilde{x})}{\partial \tilde{x}}
            &= \operatorname{col}\!\left(
                \mathbb{I}_{\tilde{n}},\;
                \frac{\partial}{\partial \tilde{x}}\!\left(\dfrac{1}{b_{jk}(\tilde{x})}\right),\;
                \frac{\partial}{\partial \tilde{x}}\!\left(\dfrac{1}{d_i(\tilde{x})}\right)
                \right) \\
            &= \widehat{P}(\hat{x}) \in \mathbb{R}[\hat{x}]^{\hat{n} \times \tilde{n}} \subset \mathbb{R}(\hat{x})^{\hat{n} \times \tilde{n}}. 
        \end{align*}
        Since $\tilde{x}$ is contained in $\hat{x}$ and $\tilde{g}_j(\tilde{x}) \in \mathbb{R}[\hat{x}]^{\tilde{n}}$, we can find functions $\hat{\xi}_j(\hat{x}) = \tilde{g}_j(\tilde{x})$. Hence, $\widehat{\Psi}$ satisfies Definition \ref{def:DA_immmersion} and yields a polynomial system. This finishes the proof. 
    \renewcommand{\qedsymbol}{$\blacksquare$}
    \end{proof}
    \renewcommand{\qedsymbol}{$\square$}
    
    \section{Structure-preserving Polynomial Immersion of Port-Hamiltonian Systems}
    \label{sec:main_results}
    Now, we shift our focus to preserving the pH structure under DA polynomial immersions. More precisely, we consider nonlinear port-Hamiltonian systems \eqref{eqn:pH_system} with a rational Hamiltonian.

    First, we introduce the notion of a lifted (DA) immersion, i.e., an immersion per Definition \ref{def:DA_immmersion} that contains the original state $x$. This extends the concept of  state lifting to systems with outputs. 
    
    \begin{definition}[Lifted immersion]
        \label{def:lifted_imm}
        A DA immersion of the form 
        \[ \Psi(x) = [x^\top \; \psi_1(x) \; \ldots \; \psi_r(x)]^\top,\, r \in \mathbb{N}, \] 
        where $\psi_1(x), \ldots, \psi_r(x)$ are DA functions, is called a \textit{lifted immersion}. \qed
    \end{definition}

    In contrast to immersions, the dimension of the target system obtained by lifted immersions is always strictly greater than that of the immersed system. 

    \begin{lemma}[Lifted immersions into rational systems]
    \label{lemma:lifted_rational_immersion}
	Given a system \eqref{eqn:affine_control_sys} whose functions are all differential algebraic, then there exists a lifted immersion of \eqref{eqn:affine_control_sys} into a rational system \eqref{eqn:rat_sys}. 
    \end{lemma}
    \begin{proof}
			By Theorem \ref{thm:rat_DA_imm} we have that
            \[
            \mathbb{R} \langle g_0,\ldots,g_m, h \rangle_{\nabla} \equiv \mathbb{R}(\psi_1(x), \ldots, \psi_{\tilde{n}}(x))
            \]
            and the immersion is given by 
            \[
            \tilde{x}:= \Psi(x)= [\psi_1(x) \; \ldots \;  \psi_{\tilde{n}}(x)]^\top.
            \]
            The corresponding rational system reads
			\begin{align*}
				\dot{\tilde{x}} &= P(\tilde{x})(\xi_0(\tilde{x}) + \sum\limits_{i=1}^{m} \xi_i(\tilde{x}) u_i) \\
				y &= \tilde{h}(\tilde{x}),
			\end{align*}
		where 
            \begin{align*}
                 &P(\tilde{x}) \in \mathbb{R}(\bar{x})^{\tilde{n} \times n},\, \tilde{h}(\tilde{x}) \in \mathbb{R}(\tilde{x})^p, \\ &P(\tilde{x})\xi_j(\tilde{x}) \in \mathbb{R}(\tilde{x})^{\tilde{n}}, \; \forall j \in \mathbb{Z}^+_m.
            \end{align*}

            It follows that $\xi_j(\tilde{x}) \in \mathbb{R}(\tilde{x})^{\tilde{n}}$.
            Consider the map 
            \[
            \overline{\Psi}: \mathcal{D} \rightarrow \mathbb{R}^{n + \tilde{n}}, \; \overline{\Psi}(x) := [x^\top \; \tilde{x}^\top]^\top
            \]
            We now show that $\overline{\Psi}(x)$ is a lifted immersion leading to a rational system. Let $\bar{x}:= \overline{\Psi}(x)$ and  $\bar{h}(\bar{x}) := \tilde{h} \circ \bar{x}_{n+1:n+\tilde{n}} \in \mathbb{R}(\bar{x})^p$. Consider 
            \begin{equation}
                \label{eqn:pbar_xtilde}
				\overline{P}(\tilde{x}) := 
                        \begin{bmatrix}
					   \mathbb{I}_n \\
					   P(\tilde{x})
				    \end{bmatrix}, \quad  
                    \bar{\xi}_j(\tilde{x}) := 
					   \xi_j(\tilde{x}), \; \forall i\in \mathbb{N}_m^0. 
		\end{equation}
            By construction, it holds that $\overline{P}(\tilde{x}) \in \mathbb{R}(\tilde{x})^{(n + \ell) \times n}$ and $\bar{\xi}_i(\tilde{x}) \in \mathbb{R}(\tilde{x})^{n}$.
            Substituting $\tilde{x}$ with $\bar{x}_{n+1:n+\tilde{n}}$ in \eqref{eqn:pbar_xtilde}, we obtain 
            \[
                \overline{P}(\bar{x}) \in \mathbb{R}(\bar{x})^{(n + \tilde{n}) \times n} \text{ and } \bar{\xi}_i(\bar{x}) \in \mathbb{R}(\bar{x})^{n}.
            \]
            Notice that $\dot{\bar{x}}$ can be written as follows
            \begin{equation}
                \label{eqn:lifted_rat_sys}
                \begin{aligned}
                    \dot{\bar{x}} &= \dfrac{\partial{\overline{\Psi}}}{\partial{x}}(x) \left(g_0(x) + \sum\limits_{i=1}^{m}g_i(x)u_i \right) \\
                    &= \overline{P}(\bar{x}) \left( \bar{\xi}_0(\bar{x}) + \sum\limits_{i=1}^{m} \bar{\xi}_i(\bar{x})u_i \right) \in \mathbb{R}(\bar{x})^{\tilde{n} +n}. 
                \end{aligned}
            \end{equation} 
            This finishes the proof.
    \renewcommand{\qedsymbol}{$\blacksquare$}
    \end{proof}
    \renewcommand{\qedsymbol}{$\square$}

    \begin{remark}[Avoiding redundancy in Lemma \ref{lemma:lifted_rational_immersion}]
    \label{rmk:redundancy}
        In Lemma \ref{lemma:lifted_rational_immersion}, we extend the full state with the functions $\psi_1(x), \ldots, \psi_{\tilde{n}}(x)$. Note that any of the $\psi_i(x), \, i \in \mathbb{N}_{\tilde{n}}$, could also be identical to any of the original states $x_1, \ldots, x_n$. Of course, this creates redundancy as a number of states might be repeated in the lifted immersion $\overline{\Psi}(x)$, hence such functions should be eliminated from $\Psi(x)$. \qed
    \end{remark}

    The following result is a natural extension of Lemma \ref{lemma:lifted_rational_immersion} to the case of polynomial systems.

    \begin{lemma}[Lifted immersions into polynomial systems]
        \label{lemma:lifted_poly}
        Given a system \eqref{eqn:affine_control_sys} whose functions are all differential algebraic, then there exists a lifted immersion of \eqref{eqn:affine_control_sys} into a polynomial system \eqref{eqn:poly_sys}. \qed
    \end{lemma}

    The proof is a straightforward application of Theorem \ref{thm:poly_imm} to Lemma \ref{lemma:lifted_rational_immersion} and thus omitted. 

    The technical assumptions for our main result are as follows.
    
    \begin{assumption}[Rational Hamiltonian]
        \label{assump:rat_Ham}
        The Hamiltonian for \eqref{eqn:pH_system} is rational in $x$, i.e., $H(x) \in \mathbb{R}(x)$. \qed
    \end{assumption}

    \begin{assumption}[Differential algebraic functions]
        \label{assump:DA_fncs_pH}
        For a port-Hamiltonian system \eqref{eqn:pH_system}, the functions 
        \[
        J_{ij}(x), \, R_{ij}(x), \, g_{kj}(x)
        \] 
        for all $i,j \in \mathbb{N}_n, \, k \in \mathbb{N}_m^0$
        are differential algebraic, i.e., they satisfy \eqref{eqn:ADE}. \qed
    \end{assumption}

    \begin{theorem}[Rational representation of pH systems]
    \label{thm:main_result}
        If the port-Hamiltonian system \eqref{eqn:pH_system} satisfies Assumptions~\ref{assump:rat_Ham} and \ref{assump:DA_fncs_pH}, then there exists a system defined on an open set $\overline{\mathcal{D}} \subseteq \mathbb{R}^{\bar{n}}, \; \overline{\Psi}(\mathcal{D}) \subseteq \overline{\mathcal{D}}$, with $\bar{n}>n$ of the form
        \vspace{-0.2em}
        \begin{subequations}
            \label{eqn:rational_rep_pH}
            \begin{align}
            \dot{\bar{x}} &= (\mathscr{J}(\bar{x}) - \mathscr{R}(\bar{x}))\nabla \overline{H}(\bar{x}) + \sum\limits_{i=1}^{m}\lambda_i(\bar{x}) u_i \label{eqn:rational_rep_pH_ODE} \\
            \bar{y} &= \operatorname{col}(\lambda_1^\top(\bar{x}),\ldots,\lambda_m^\top(\bar{x}))\nabla \overline{H}(\bar{x})
        \end{align}
        \end{subequations}
        such that
    \begin{enumerate}[label=(\roman*), itemsep=0pt]
        \item Lifting: $\bar{x}= \overline{\Psi}(x)$ is a lifted state of $x$.
        \item Invariance: The set $\overline{\Psi}(\mathcal{D})$ defines an invariant embedded submanifold of $\mathbb{R}^{\bar{n}}$.
        \item Immersion: For any $x \in \mathcal{D}$ there exists a corresponding $\bar{x}$ such that $\bar{y} = y$.
        \item Rational structure: The obtained system is a rational system in $\bar{x}$.
        \item Interconnection geometry preservation: The matrix $\mathscr{J}(\bar{x}) = \operatorname{diag}(\bar{J}(\bar{x}), 0)$ is skew-symmetric with $\overline{J}(\bar{x}) = J(x)$.
        \item Energy preservation: For all $x \in \mathcal{D}$, there exists a corresponding $\bar{x}$ such that $\overline{H}(\bar{x}) = H(x)$.
        \item Preservation of dissipation:
        The matrix $\mathscr{R}(\bar{x})$ is symmetric and satisfies for all $\bar{x} \in \overline{\Psi}(\mathcal{D})$ and $x \in \mathcal{D}$ the following
        \[
            (\nabla \overline{H}(\bar{x}))^\top \mathscr{R}(\bar{x}) \nabla \overline{H}(\bar{x}) = (\nabla H(x))^\top R(x) \nabla H(x)\ge 0.
        \]
           \qed
    \end{enumerate} 
    \end{theorem}
        
    \begin{proof}
        The proof proceeds as follows: In Part 1, we first utilize Lemma \ref{lemma:lifted_rational_immersion} to construct a lifted immersion into a rational system, hence addressing (i)-(iv) of the result. The fact that $\Psi(\mathcal{D})$ defines an $n$-dimensional embedded submanifold of $\mathbb{R}^{\bar{n}}$ holds, since it is the graph of a smooth map over an open set (see, e.g., Proposition 5.4 in \cite{lee2003smooth}). In Part 2, by preserving the Hamiltonian and lifting the internal interconnection matrix of the immersed system into the higher dimensional space, we obtain the structure obtained in \eqref{eqn:rational_rep_pH_ODE}. This also addresses (v) and the energy preservation property mentioned in (vi). In Part 3, we then show that (vii) follows. Finally, in Part 4, we show that the output can be represented in the pH form \eqref{eqn:ODE_output}.
        
        \paragraph{Part 1} Since differential algebraic functions are closed under multiplication, addition, and subtraction (see Remark \ref{rmk:DA_fncs}), Assumption \ref{assump:DA_fncs_pH} implies that the functions
        \[
            g_0(x):= (J(x)-R(x)) \nabla H(x) \text{ and } g_i(x), \, i \in \mathbb{N}_m
        \]
        are also differential algebraic. By Lemma \ref{lemma:lifted_rational_immersion}, we can construct a lifted immersion for \eqref{eqn:pH_system} into a rational system of the form \eqref{eqn:lifted_rat_sys} with
        \begin{align*}
            g_0(x) &= \bar{\xi}_0(\bar{x}) \in \mathbb{R}(\bar{x})^{n} \\
            g_i(x) &= \bar{\xi}_i(\bar{x}) \in \mathbb{R}(\bar{x})^{n}, \; i \in \mathbb{N}_m.
        \end{align*}
        This addresses (i-iv).
        
         \paragraph{Part 2} From Assumption \ref{assump:rat_Ham}, it follows that $\nabla H(x) \in \mathbb{R}(x)^n$. Replacing $x$ with $\bar{x}_{1:n}$, we obtain $\nabla_x H \circ \bar{x}_{1:n} \in \mathbb{R}(\bar{x})^n$. Thus, there exist matrices $\overline{J}(\bar{x}), \, \overline{R}(\bar{x}) \in \mathbb{R}(\bar{x})^{n \times n}$ such that $\overline{J}(\bar{x}) = J(x)$ and $\overline{R}(\bar{x}) = R(x)$. Hence, we have 
            \begin{equation}
                \label{eqn:main_res_rat}
                \dot{\bar{x}} = 
				\overline{P}(\bar{x})
                    \left( 
					(\overline{J}(\bar{x}) - \overline{R}(\bar{x})) \nabla_x H(\bar{x}) \\ 
				    + \sum\limits_{i=1}^{m} \bar{\xi}_i(\bar{x})u_i \right).
            \end{equation}
        Recall that, by construction, a lifted immersion contains the original state $x$. To preserve the energy of \eqref{eqn:pH_system} in the rational system, we define the Hamiltonian for \eqref{eqn:main_res_rat} as
        \[
            \overline{H}(\bar{x}) := H \circ \bar{x}_{1:n} \in \mathbb{R}(\bar{x}).
        \]
        This addresses (vi).
        This implies
        \[
            \nabla_{\bar{x}} \overline{H}(\bar{x}) = 
            \begin{bmatrix}
                \nabla_x H^\top \circ \bar{x}_{1:n} & 0^\top   
            \end{bmatrix}^\top.
        \]
        Rearranging \eqref{eqn:main_res_rat} we arrive at
            \begin{multline}
                \label{eqn:pseudo_pH}
                \dot{\bar{x}} = \Bigg(
                \underbrace{
                \begin{bmatrix}\,
                    \overline{J}(\bar{x}) & 0 \\
                    0 & 0
                \end{bmatrix}}_{\mathscr{J}(\bar{x})} - 
                    \underbrace{
                    \begin{bmatrix}\,
                    \overline{R}(\bar{x}) & \Lambda^\top(\bar{x}) \\
                    \Lambda(\bar{x}) & 0
                    \end{bmatrix}}_{\mathscr{R}(\bar{x})}
                \Bigg)
                \nabla \overline{H}(\bar{x}) \\ +
                \begin{bmatrix}
                    \sum\limits_{i=1}^{m} \bar{\xi}_i(\bar{x}) u_i \\
                    \sum\limits_{i=1}^{m} P(\bar{x}) \bar{\xi}_i(\bar{x}) u_i
                \end{bmatrix},
            \end{multline}  
            where 
            $\Lambda(\bar{x}) = {P}(\bar{x})(\overline{R}(\bar{x}) - \overline{J}(\bar{x}))$ and
            \[
                \lambda_i = 
                \begin{bmatrix}
                    \bar{\xi_i}(\bar{x}) \\
                    P(\bar{x}) \bar{\xi_i}(\bar{x})
                \end{bmatrix}, \; i \in \mathbb{N}_m.
            \] 
            
            Observe that the rational matrix $\mathscr{J}(\bar{x})$ is skew-symmetric, hence a valid internal interconnection matrix. In fact, it is the original internal interconnection matrix (in new coordinates) \say{lifted} into a higher dimensional space. This addresses (v).
            
            The matrix $\mathscr{R}(\bar{x})$ is symmetric by construction, however, it is not necessarily positive semi-definite even on $\overline{\Psi}(\mathcal{D})$. This is because a necessary condition for the positive semi-definiteness of $\mathscr{R}(\bar{x})$ to hold is that $\Lambda(\bar{x})=0$ for all $\bar{x} \in \overline{\mathcal{D}}$ which clearly would not hold unless the original system has no drift. Further investigation is out of scope and we refer to \cite{albert1969conditions,bekker1988positive}. It is worth noting that the $0$ block matrix in $\mathscr{R}(\bar{x})$ is not unique and might be replaced by any symmetric matrix without changing the result. Hence, \eqref{eqn:pseudo_pH} is not a unique representation. 
            
            \paragraph{Part 3} By direct calculation, it follows for all $x \in \mathcal{D}$: 
                \begin{equation*}
                    \nabla \overline{H}^\top(\bar{x}) \mathscr{R}(\bar{x}) \nabla \overline{H}(\bar{x})= \\ \nabla H^\top(x) R(x) \nabla H(x) \geq 0.
                \end{equation*} 
                Put differently, $\mathscr{R}(\bar{x})$ satisfies a dissipation inequality on $\overline{\Psi}(\mathcal{D})$. 
            Since $\overline{\Psi}(x)$ is a lifted immersion, which is injective by construction, hence
            \[
                \label{eqn:prop_Rtilde}
                    \nabla^\top \overline{H}(\bar{x}) \mathscr{R}(\bar{x}) \nabla \overline{H}(\bar{x}) \geq 0, \; \forall \bar{x} \in \overline{\Psi}(\mathcal{D}),
            \] 
            which corresponds to (vii).
            
           \paragraph{Part 4} It remains to prove that the natural output (in the sense of port-Hamiltonian systems) corresponding to \eqref{eqn:pseudo_pH} is identical to the output of the original system. Direct calculation yields  
             \begin{align*}
                \begin{bmatrix}
                    \bar{\xi}_1^\top(\bar{x}) & \bar{\xi}_1^\top(\bar{x})P^\top(\bar{x}) \\
                    \vdots & \vdots \\
                    \bar{\xi}_m^\top(\bar{x}) & \bar{\xi}_1^\top(\bar{x})P^\top(\bar{x})
                \end{bmatrix}^\top \nabla_{\bar{x}} \overline{H}(\bar{x}) &=  
                \begin{bmatrix}
                    \bar{\xi}_1^\top(\bar{x}) \\ 
                    \vdots \\
                    \bar{\xi}_m^\top(\bar{x})
                \end{bmatrix} 
                \nabla H (\bar{x}_{1:n}) \\
                &= 
                \begin{bmatrix}
                    g_1^\top(x) \\ \vdots \\ g_m^\top(x)
                \end{bmatrix} 
                \nabla H (x) = y. 
			\end{align*} 
            This finishes the proof.
    \renewcommand{\qedsymbol}{$\blacksquare$}
    \end{proof}
    \renewcommand{\qedsymbol}{$\square$}
    
    \begin{lemma}[Passivity preservation] 
            The rational port-Hamiltonian system \eqref{eqn:pseudo_pH} is passive on $\overline{\Psi}(\mathcal{D})$ with storage function $\overline{H}(\bar{x})$.
    \end{lemma}

    \begin{proof} 
        \begin{align*}
            \dfrac{\mathrm{d} \overline{H}(\bar{x})}{\mathrm{d}t} &= \nabla \overline{H}^\top(\bar{x}) \; \dot{\bar{x}} \\
            & = \nabla \overline{H}^\top(\bar{x}) \mathscr{J}(\bar{x}) \nabla \overline{H}(\bar{x}) - \nabla \overline{H}^\top(\bar{x}) \mathscr{R}(\bar{x}) \nabla \overline{H}(\bar{x}) + \\ 
            & \qquad \qquad  \nabla^\top \overline{H}(\bar{x})[\lambda_1(\bar{x}) \; \ldots \; \lambda_m(\bar{x})]u. 
        \end{align*}
        The first term vanishes due to skew-symmetry yielding 
        \[
            \dfrac{d \overline{H}(\bar{x})}{dt} = -\nabla \overline{H}^\top(\bar{x}) \mathscr{R}(\bar{x}) \nabla \overline{H}(\bar{x}) + y^\top u \leq u^\top y. 
        \]
        This finishes the proof.
        \renewcommand{\qedsymbol}{$\blacksquare$}
    \end{proof}
    \renewcommand{\qedsymbol}{$\square$}

    \begin{theorem}[Polynomial representation of pH systems]
    \label{cor:main_corollary}
    The rational system \eqref{eqn:rational_rep_pH} obtained in Theorem \ref{thm:main_result} can be immersed into a polynomial system of the same structure.   \qed
    \end{theorem}

    \begin{proof}
    The proof follows by applying the steps given in the proofs of Theorems \ref{thm:poly_imm} and \ref{thm:main_result} respectively to \eqref{eqn:pseudo_pH}. 
    \renewcommand{\qedsymbol}{$\blacksquare$}
    \end{proof}
    \renewcommand{\qedsymbol}{$\square$}
        
    \section{Illustrative Examples}
    We present two examples to elaborate on the result obtained in the previous section. The first example is a purely academic example, while the second example is the ubiquitous rolling coin example. 
    \label{sec:examples}
    \subsection{Example with Exponential Terms}
        Consider the following port-Hamiltonian system 
        \begin{equation}
            \label{eqn:pH_example}
            \begin{aligned}
            \dot{x} &= \big(J - R(x)\big)\nabla H(x) + g\,u, \\
            y &= g^\top \nabla H(x),
            \end{aligned}
            \end{equation}
            where $x \in \mathbb{R}^2, \; H(x) = \frac{1}{2}(x_1^2 + x_2^2)$ and
            \[
            J =
            \begin{bmatrix}
            0 & 1 \\[2pt]
            -1 & 0
            \end{bmatrix}, \quad
            R(x) =
            \begin{bmatrix}
            e^{x_1} & 0 \\[2pt]
            0 & e^{x_2}
            \end{bmatrix}, \quad
            g =
            \begin{bmatrix}
            1 \\[2pt]
            1
            \end{bmatrix}.
            \]
        Since the Hamiltonian is quadratic and all functions of \eqref{eqn:pH_example} are differential algebraic, then we can find a rational system of the form \eqref{eqn:rational_rep_pH}. The differential field corresponding to \eqref{eqn:pH_example} is 
        \[
        \mathbb{R} \langle x_2-x_1 e^{x_1}, -x_1-x_2 e^{x_2}, x_1+x_2 \rangle_\nabla
        \]
        which is identical to $\mathbb{R}(x_1,x_2,e^{x_1},e^{x_2})$ which in turn is finitely generated. Hence, we consider the following (lifted)\footnote{Note that the generators of $\mathbb{R}(x_1,x_2,\exp{x_1}, \exp{x_2})$ already contain the state $x$, hence the immersion obtained by Theorem \ref{thm:rat_DA_imm} is inherently a lifted immersion.} immersion
        \[
            \bar{x} := \overline{\Psi}(x) = [x_1 \;\;  x_2 \;\; e^{x_1} \;\; e^{x_2}]^\top, 
        \] 
        and the preserved Hamiltonian $\overline{H}(\bar{x}) := \frac{1}{2}(\bar{x}_1^2 + \bar{x}_2^2)$. Now the matrices in the proof of Theorem \ref{thm:main_result} are
        \[
        \overline{P}(\bar{x}) = \overline{R}(\bar{x}) =
        \begin{bmatrix}
            \bar{x}_3 & 0 \\
            0 & \bar{x}_4
        \end{bmatrix},
        \quad
        \Lambda(\bar{x}) = 
        \begin{bmatrix}
            \bar{x}_3^2 & -\bar{x}_3 \\
            \bar{x}_4 & \bar{x}_4^2
        \end{bmatrix}.
        \]
        It is also worth noting that $\overline{J} = J$ and $\bar{\xi}_1 = \bar{\xi}_2 = 1$. We immediately obtain the the following polynomial system 
    \begin{equation}
        \label{eqn:ph_barx}
        \begin{aligned}
        \dot{\bar x} &= \big(\mathscr{J}(\bar x)-\mathscr{R}(\bar x)\big)\,\nabla \overline H(\bar x) + \lambda(\bar x)\,u,\\
        y &= \lambda(\bar x)^\top \nabla \overline H(\bar x),
        \end{aligned}
    \end{equation}
    with $\overline H(\bar x)=\tfrac12\big(\bar{x}_1^2+\bar x_2^2\big)$, $\mathscr{J}(\bar{x}) = \operatorname{diag}(\overline{J}, 0_2)$,
    \[
    \mathscr{R}(\bar x)=
    \begin{bmatrix}
    \bar x_3 & 0 & \bar x_3^{2} & \bar x_4\\[2pt]
    0 & \bar x_4 & -\bar x_3 & \bar x_4^{2}\\[2pt]
    \bar x_3^{2} & -\bar x_3 & 0 & 0\\[2pt]
    \bar x_4 & \bar x_4^{2} & 0 & 0
    \end{bmatrix} \quad
    \text{ and } \quad
    \lambda(\bar x)=
    \begin{bmatrix}
    1\\[2pt]
    1\\[2pt]
    \tilde x_3\\[2pt]
    \tilde x_4
    \end{bmatrix}.
    \]
    The dissipated power for system \eqref{eqn:ph_barx} reads
    \[
    \nabla \overline{H}^\top(\bar{x}) \mathscr{R}(\bar{x}) \nabla \overline{H}(\bar{x}) = \bar{x}_1^2 \bar{x}_3 + \bar{x}_2^2 \bar{x}_4.
    \]
    Observe that, for $\bar{x} \in \overline{\Psi}(\mathbb{R}^2)$, this matches exactly the dissipated power for \eqref{eqn:pH_example} which reads
    \[
    \nabla H^\top(x) R(x) \nabla H(x) = x_1^2 e^{x_1} + x_2^2 e^{x_2} \geq0, \forall x \in \mathbb R^2.
    \]
    Thus, 
    \[
        \nabla \overline{H}^\top(\bar{x}) \mathscr{R}(\bar{x}) \nabla \overline{H}(\bar{x}) \geq 0, \; \forall \bar{x} \in \overline{\Psi}(\mathbb{R}^2)
    \]
    and passivity of \eqref{eqn:ph_barx} follows immediately on $\overline{\Psi}(\mathbb{R}^2)$. 
    
    \subsection{The rolling coin}
    \begin{figure}[htb]
        \centering
        \includegraphics[scale=0.7]{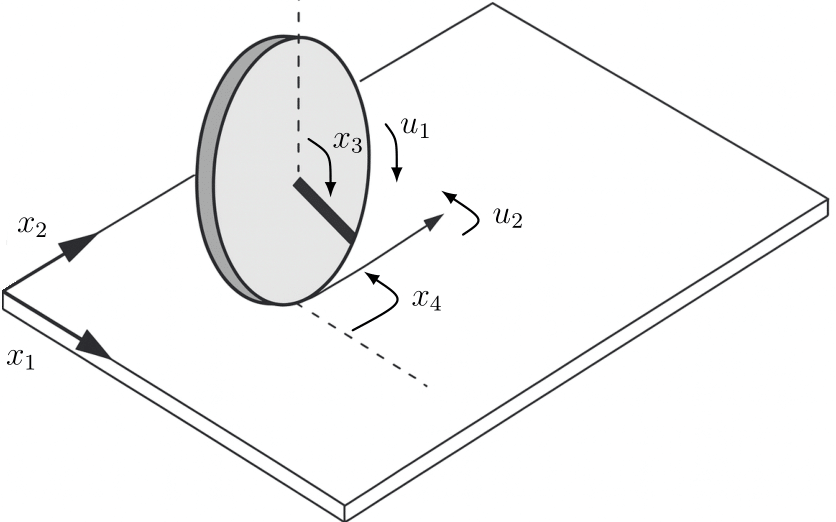}
        \caption{A schematic of the rolling coin.}
        \label{fig:rolling_coin}
    \end{figure}
    Consider the port-Hamiltonian model of a rolling coin on a horizontal plane (Figure \ref{fig:rolling_coin}) \cite{duindam2009modeling}
        \begin{equation}
            \begin{aligned}
            \label{eqn:pH_euro}
            \dot{x} &= J(x) \nabla H(x) + 
            G 
            \begin{bmatrix}
                u_1 \\
                u_2
            \end{bmatrix} \\
            y &= G^\top \nabla H(x) = 
            \begin{bmatrix}
                \frac{1}{2}x_6 & x_5
            \end{bmatrix}^\top,
            \end{aligned}
        \end{equation}
           where 
        \begin{align*}
        J(x) &= 
            \begin{bmatrix}
                0 & 0 & 0 & 0 & 0 & \cos{x_4} \\
                0 & 0 & 0 & 0 & 0 & \sin{x_4} \\
                0 & 0 & 0 & 0 & 0 & 1 \\
                0 & 0 & 0 & 0 & 1 & 0 \\
                0 & 0 & 0 & -1 & 0 & 0 \\
                -\cos{x_4} & -\sin{x_4} & -1 & 0 & 0 & 0 
            \end{bmatrix} \\
        G &= 
        \begin{bmatrix}
            g_1 & g_2 
        \end{bmatrix} = 
        \left[
            \begin{array}{cc}
            \multicolumn{2}{c}{0_{4\times 2}} \\
                0 & 1\\
                1 & 0
            \end{array}
        \right], \quad H(x) = \dfrac{1}{2}x_5^2 + \dfrac{1}{4} x_6^2. 
        \end{align*}
        
        The states $x_1, x_2$ denote the position of the point of contact of the coin with the $x_1x_2$-plane, $x_3$ denotes the angle of the marked segment on the coin, $x_4$ denotes the heading angle, and $x_5$ denotes the angular momenta with respect to $x_4$. The state $x_6$ is a nonlinear function involving the momenta with respect to $x_1,x_2,x_3$ and the heading angle $x_4$. The inputs $u_1,u_2$ represent the torques about the rolling axis and the vertical axis respectively. All constants of the model are assumed to be one.

        This example satisfies Assumptions \ref{assump:rat_Ham} and \ref{assump:DA_fncs_pH}. That is, respectively, $H(x)$ is polynomial (hence rational) in $x$ and all functions are differential algebraic functions satisfying $\eqref{eqn:ADE}$. We first construct a rational immersion according to Theorem \ref{thm:rat_DA_imm}. To this end, consider the differential field corresponding to \eqref{eqn:pH_euro}
        \[
            \mathbb{R}\langle x_6\cos{x_4}, x_6\sin{x_4}, x_5, x_6 \rangle_\nabla 
        \]
        which is identical to $\mathbb{R}(x_5,x_6,\cos{x}_4, \sin{x}_4)$
         as a non-differential field. The latter is finitely generated, then the function 
         \[
         \Psi(x) := 
         \begin{bmatrix}
             x_5 & x_6 & \cos{x}_4 & \sin{x}_4
         \end{bmatrix}^\top
         \]
         defines an immersion into a rational (in fact a polynomial) system. Removing redundant elements per Remark \ref{rmk:redundancy} we redefine $\Psi(x) := [\cos{x}_4 \; \sin{x}_4]^\top$. Consider the lifted immersion according to Lemma \ref{lemma:lifted_rational_immersion} which reads 
        \[
        \bar{x}:= \overline{\Psi}(x) = [x^\top \;\; \cos{x_4} \;\; \sin{x_4}]^\top.
        \]
       The preserved Hamiltonian from Theorem \ref{thm:main_result} reads
        \[
            \overline{H}(\bar{x}) = \dfrac{1}{2}\bar{x}_5^2 + \dfrac{1}{4} \bar{x}_6^2.
        \]
        We compute the Jacobian as
            \begin{align*}
                \operatorname{D}\Psi(x) 
                &= 
                \begin{bmatrix}
                    -\sin{x}_4 & 0_{1 \times 5} \\
                    \cos{x}_4  & 0_{1 \times 5}
                \end{bmatrix} \\
                &= 
                \begin{bmatrix}
                    -\bar{x}_8 & 0_{1 \times 5} \\
                    \bar{x}_7  & 0_{1 \times 5}
                \end{bmatrix} = P(\bar{x}) \in \mathbb{R}[\bar{x}]^{2\times6}.
            \end{align*}
        The internal interconnection matrix reads 
        \[
        \mathscr{J}(\bar{x}) =
        \begin{bmatrix}
            0 & 0 & 0 & 0 & 0 & \bar{x}_7 & 0 & 0 \\
            0 & 0 & 0 & 0 & 0 & \bar{x}_8 & 0 & 0 \\
            0 & 0 & 0 & 0 & 0 & 1 & 0 & 0 \\
            0 & 0 & 0 & 0 & 1 & 0 & 0 & 0 \\
            0 & 0 & 0 & -1 & 0 & 0 & 0 & 0 \\
            -\bar{x}_7 & -\bar{x}_8 & -1 & 0 & 0 & 0 & 0 & 0 \\
            0 & 0 & 0 & 0 & 0 & 0 & 0 & 0 \\
            0 & 0 & 0 & 0 & 0 & 0 & 0 & 0 
        \end{bmatrix}
        \]
        and the new control vector fields read 
        \[
            \bar{\xi}_1(\bar{x}) = [0_{1\times5} \;\; 1 \;\; 0_{1\times2}]^\top, \; \bar{\xi}_2(\bar{x}) = [0_{1\times4} \;\; 1 \;\; 0_{1\times3}]^\top.
        \] 
        The natural output corresponding to \eqref{eqn:pseudo_pH} reads 
        \[ 
            \begin{bmatrix}
               \bar{\xi}_1^\top(\bar{x}) & P(\bar{x})\bar{\xi}_1^\top(\bar{x}) \\ \bar{\xi}_2^\top(\bar{x}) & P(\bar{x})\bar{\xi}_2^\top(\bar{x})
            \end{bmatrix}^\top
            \nabla_{\bar{x}}{\overline{H}}(\bar{x}) =
            \begin{bmatrix}
                G_1^\top \nabla H(x) \\
                G_2^\top \nabla H(x)
            \end{bmatrix} = y.
        \]
        With $\overline{R}(\bar{x}) = 0_6$, we  compute $\Lambda(\bar{x})$ as
        \[
            \Lambda(\bar{x}) = 
            \begin{bmatrix}
                0_{1\times5} & \bar{x}_7 \bar{x}_8 \\
                0_{1\times5} & -\bar{x}_7^2
            \end{bmatrix}.
        \]
        The new resistive matrix reads 
        \[
        \mathscr{R}(\bar{x}) = 
            \begin{bmatrix}
                0_{6\times5} & 0_{6\times1} & 0_{6\times1} & 0_{6\times1} \\
                0_{1\times5} & 0 & \bar{x}_7\bar{x}_8 & -\bar{x}_7^2 \\
                0_{1\times5} & \bar{x}_7 \bar{x}_8 & 0 & 0 \\
                0_{1\times5} & -\bar{x}_7^2 & 0 & 0 \\
            \end{bmatrix}
        \]
        The set of eigenvalues for $\mathscr{R}(\bar{x})$ is 
        \[
            \sigma(\mathscr{R}(\bar{x})) = \left\{0,\bar{x}_7\sqrt{\bar{x}_7^2 + \bar{x}_8^2},-\bar{x}_7\sqrt{\bar{x}_7^2 + \bar{x}_8^2} \right\},
        \]
        where the zero has an algebraic multiplicity equal to $6$. The matrix $\mathscr{R}(\bar{x})$ is not positive semi-definite even on $\overline{\Psi}(\mathcal{D})$ as we have $\bar{x}_7, \bar{x}_8 \in [-1,1]$. Similar to \eqref{eqn:pH_euro}, the obtained polynomial  system is conservative on $\overline{\Psi}(\mathcal{D})$ (in fact on $\mathbb{R}^8$) since 
        \[
        \nabla^\top \overline{H}(\bar{x}) \mathscr{R}(\bar{x}) \nabla^\top \overline{H}(\bar{x}) = 0. 
        \]  
        The state and output trajectories of \eqref{eqn:pH_euro} and the polynomial port-Hamiltonian system are shown in Figure \ref{fig:euro_sim} for $u_i(t) = 0.5\sin(2t), \; i =1,2$, and $x_0 = 0_{6\times1}, \bar{x}_0 = \overline{\Psi}(x_0) = [0_{1\times6} \;\; 1 \;\; 0]^\top$.  
%---------------------------
        \begin{figure}
            \centering
            \includegraphics[scale=0.25]{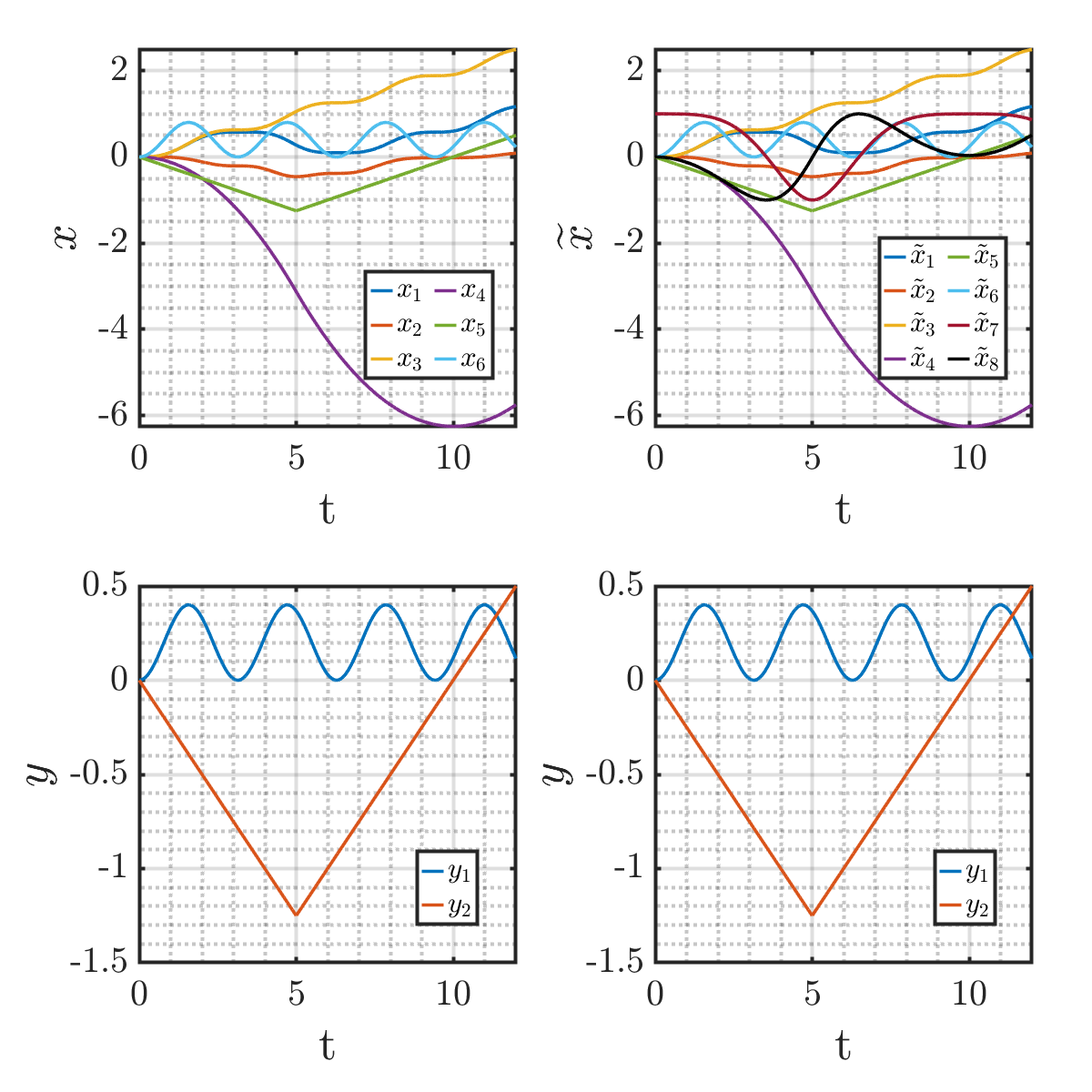}
            \caption{State and output trajectories: Left = system \eqref{eqn:pH_euro}, right = polynomial port-Hamiltonian system.}
            \label{fig:euro_sim}
        \end{figure}
%--------------------
\section{SOS-IDA-PBC via Polynomial Immersions} \label{sec:ida_sos}
In this section, we illustrate how polynomial immersions facilitate control design. In particular, we consider Interconnection and Damping Assignment Passivity-Based Control (IDA-PBC; \cite{ortega2002interconnection}), which is a well-established control technique for pH systems. In this method, the feedback law relies on the solution of partial differential equations called the matching equations. For the sake of brevity, we present here, in terms of an example, the idea of combining IDA-PBC with the computationally efficient SOS optimization. Generalization and rigorous analysis are subject of future work.

Consider the pH system
\begin{align*}
    \dot{x} &= -
    \begin{bmatrix}
        e^{x_1} & 0 \\
        0 & e^{x_2}
    \end{bmatrix}
    \begin{bmatrix}
        x_1 \\
        x_2
    \end{bmatrix} + 
    \begin{bmatrix}
        0 \\
        1
    \end{bmatrix}u \\
    y &= x_2,
\end{align*}
where the Hamiltonian of the system is $H(x)= \frac{1}{2}(x_1^2+x_2^2)$. We first construct an immersion according to Theorem \ref{thm:main_result}. We start by defining the following immersion 
\[
    \bar{x} := \overline{\Psi}(x) = 
    \begin{bmatrix}
        x^\top & e^{x_1} & e^{x_2}
    \end{bmatrix}^\top.
\]
Proceeding as in the proof of Theorem \ref{thm:main_result}, we obtain 
\begin{equation}
    \label{eqn:toy_example_target}
    \begin{aligned}
        \dot{\bar{x}} &= -
        \begin{bmatrix}
            \bar{x}_3   &  0            &  \bar{x}_3^2 & 0 \\
            0           &  \bar{x}_4    &  0           & \bar{x}_4^2 \\
            \bar{x}_3^2 &  0            &  0           & 0 \\
            0           &  \bar{x}_4^2  &  0           & 0
        \end{bmatrix}
        \begin{bmatrix}
            \bar{x}_1 \\
            \bar{x}_2 \\
            0 \\
            0
        \end{bmatrix} + 
        \begin{bmatrix}
            0 \\
            1 \\
            0 \\
            \bar{x}_4
        \end{bmatrix}u \\
        y &= \bar{x}_2.
    \end{aligned}
\end{equation}
%-----------------------------------------------
Here, we are interested in designing a stabilizing controller for \eqref{eqn:toy_example_target} via the IDA-PBC methodology \cite{ortega2002interconnection}. IDA-PBC relies on 
the so-called matching equations
\begin{multline}
    \label{eqn:IDA-PBC_poly}
    \lambda_i(\bar{x})^\perp (\mathscr{J}(\bar{x})-\mathscr{R}(\bar{x})) \nabla \overline{H}(\bar{x})  
    = \\ \lambda_i(\bar{x})^\perp (\mathscr{J}_d(\bar{x}) - \mathscr{R}_d(\bar{x})) \nabla \overline{H}_d(\bar{x}), \; i \in \mathbb{N}_m.
\end{multline}

The feedback satisfying \eqref{eqn:IDA-PBC_poly} is then computed
\[
    \label{eqn:IDA-PBC_control_poly}
    u = \lambda^\dagger(\bar{x})(\mathscr{F}_d(\bar{x}) \nabla \overline{H}_d(\bar{x}) - \mathscr{F}(\bar{x}) \nabla \overline{H}(\bar{x})),
\]
where 
    \begin{align*}
        \mathscr{F}(\bar{x}) = \mathscr{J}(\bar{x}) - \mathscr{R}(\bar{x}), \; \mathscr{F}_d(\bar{x}) = \mathscr{J}_d(\bar{x}) - \mathscr{R}_d(\bar{x}), \\
        \lambda^\dagger(\bar{x})= \left( 
        \begin{bmatrix}
            \lambda_1(\bar{x}) \\
            \vdots \\
            \lambda_m(\bar{x})
        \end{bmatrix}
        \begin{bmatrix}
            \lambda_1(\bar{x}) &
            \ldots &
            \lambda_m(\bar{x})
        \end{bmatrix} \right)^{-1}
        \begin{bmatrix}
            \lambda_1(\bar{x}) \\
            \vdots \\
            \lambda_m(\bar{x})
        \end{bmatrix}.
    \end{align*}

Here, we choose
\begin{align*}
    \mathscr{J}_d(\bar{x}) &= 0, \; \mathscr{R}_d(\bar{x}) = r\cdot \label{}\mathbb{I}_{\bar{n}}, \, r \in \mathbb{R}^+ \setminus\{0\}\\ 
    \lambda^\perp &= 
    \begin{bmatrix}
        1 & 0 & 0 & 0
    \end{bmatrix}.
\end{align*}
Solving \eqref{eqn:IDA-PBC_poly} for system \eqref{eqn:toy_example_target} yields the following desired Hamiltonian
\[
     \overline{H}_d(\bar{x}) = \dfrac{1}{2r}\bar{x}_1^2 \bar{x}_3 + \omega(\bar{x}_2,\bar{x}_3,\bar{x}_4),
\]
where $\omega$ is a function to be designed such that $\overline{H}_d(\bar{x})$ defines a Lyapunov function for the closed-loop system with respect to a desired set-point $x^{\hspace{100\mu}d} \in \mathbb{R}^{\bar{n}}$. We choose for $\omega$ a polynomial ansatz of degree $4$, i.e., 
\[
    \omega
= \sum_{i+j+k \le 4} a_{ijk} (\bar{x}_2 - \bar{x}^{\hspace{100\mu}d}_2)^{\,i} \, (\bar{x}_3 - \bar{x}^{\hspace{100\mu}d}_3)^{\,j} \, (\bar{x}_4 - \bar{x}^{\hspace{100\mu}d}_4)^{\,k}.
\] 
Since a Lyapunov function must equate to zero at $x^{\hspace{100\mu}d}$, this constrains $\bar{x}^{\hspace{100\mu}d}_1$ to be zero since $\bar{x}_3 \neq 0$ for all $\bar{x} \in \overline{\Psi}(\mathbb{R}^2)$. Hence, we set $\bar{x}^{\hspace{100\mu}d} = (0,\bar{x}^{\hspace{100\mu}d}_2, 1, e^{\bar{x}^{\hspace{100\mu}d}_2}) \in \overline{\Psi}(\mathbb{R}^2)$ for some $\bar{x}^{\hspace{100\mu}d}_2 \in \mathbb{R}$. 
 
To design $\omega$ we solve the SOS program 
\begin{subequations}\label{eqn:SOS_program}
\begin{align}
\underset{a_{ijk}}{\text{min}} \quad & 1 \\
\text{s.t.} \quad & \nonumber \\
    &\nabla \overline{H}_d(\bar{x}^{\hspace{100\mu}d}) = 0 \label{req:stationarity} \\
    &\overline{H}_d(\bar{x}^{\hspace{100\mu}d}) = 0 
    \label{req:Hd_zero} \\ 
    &\nu^\top \nabla^2 \overline{H}_d(\bar{x}) \nu - \nabla \, \nu^\top \nu \;\; \text{is SOS}. \label{req:convexity}
\end{align}
\end{subequations}
It can be shown that the third constraint ensures the strict convexity of $\overline{H}_d$. Hence, \eqref{req:convexity} combined with \eqref{req:stationarity} is sufficient for $x^{\hspace{100\mu}d}$ to be a global minimizer of $\overline{H}_d$. Moreover, all the constraints combined \eqref{req:stationarity}-\eqref{req:convexity} ensure the positive definiteness of $\overline{H}_d$.

Nevertheless, enforcing \eqref{req:convexity} globally leads to infeasibility of the optimization problem. Hence, we restrict this constraint to hold locally on a closed ball $\mathcal{B}_{\tilde{r}}(\bar{x}^{\hspace{100\mu}d}) \subset \overline{\Psi}(\mathbb{R}^2)$ of radius $\tilde{r}=5$ centered at $\bar{x}^{\hspace{100\mu}d}$. To this end, we define the following functions 
\begin{align*}
    \mu_0 &= \tilde{r}^2 - (\bar{x} - x^{\hspace{100\mu}d})^\top (\bar{x} - x^{\hspace{100\mu}d}) \\
    \mu_1 &= \bar{x}_{3} - \sum_{i=0}^{3} \dfrac{e^{\bar{x}^{\hspace{100\mu}d}_1}}{i!}(\bar{x}_1 - \bar{x}^{\hspace{100\mu}d}_1)^i \\
    \mu_{2} &= \bar{x}_{4} - \sum_{i=0}^{3} \dfrac{e^{\bar{x}^{\hspace{100\mu}d}_2}}{i!}(\bar{x}_2 - \bar{x}^{\hspace{100\mu}d}_2)^i
\end{align*}
where $\mu_1, \, \mu_2$ are, respectively, Taylor approximations of the constraints defining the embedded submanifold $\overline{\Psi}(\mathbb{R}^2) \subset \mathbb{R}^4$. These constraints read 
\[
    \bar{x}_3 - e^{\bar{x}_1} = \bar{x}_4 - e^{\bar{x}_2} = 0.
\]
Notice that $\mathcal{B}_{\tilde{r}}(\bar{x}^{\hspace{100 \mu}d})$ is approximated by the polynomial inequality $\mu_0 \geq0$ and the polynomial equalities $\mu_1 = \mu_2 = 0$. By an extension of the S-Procedure \cite{papachristodoulou2005tutorial}, the local strict convexity constraint reads
\begin{subequations}
\label{eq:sos_requirements}
\begin{align}
    \label{req:sos_req_pd_hess_loc}
    \nu^\top \nabla^2 \overline{H}_d(\bar{x}) \nu - \delta \, \nu^\top \nu - \sum_{i=0}^{2} s_i \, \mu_i \quad & \text{is SOS} \\
    \label{req:s0_sos}
    s_0 \quad & \text{is SOS}
\end{align}
\end{subequations}
where $s_\tau$ denotes the polynomial multipliers defined as follows 
\[
    s_\tau = \sum_{i+\ldots+\ell \le 3} \tilde{a}_{\tau, ijk\ell}\, \bar{x}_1^{\,i} \ldots \bar{x}_{\bar{n}}^{\,j}\, \nu_1^{\,k}\ldots \nu_{\bar{n}}^{\,\ell}, \; \forall \tau \in \mathbb{N}^0_{2}.
\]
Solving the new SOS program, with \eqref{eq:sos_requirements} in place of \eqref{req:convexity} and $x_{2}^{\hspace{100 \mu}d}= 4$, yields the following Hamiltonian
\begin{align*}
\overline{H}_d(\bar{x}) = \, & 103877.32 - 1002.31\bar{x}_2 - 767.95\bar{x}_3 - 5678.35\bar{x}_4 \\
& - 37.96\bar{x}_2^2 + 373.41\bar{x}_3^2 + 113.78\bar{x}_4^2 \\
& + 14.83\bar{x}_2\bar{x}_3 + 35.46\bar{x}_2\bar{x}_4 + 42.54\bar{x}_3\bar{x}_4 \\
& + 17.29\bar{x}_2^3 - 0.99\bar{x}_4^3 - 6.22\bar{x}_2^2\bar{x}_3 + 0.59\bar{x}_2^2\bar{x}_4 \dots \\
& + (\text{higher order terms}).
\end{align*}
It is worth noting that, in addition to the coefficients $a_{ijk}$, the coefficients of the multipliers, $\tilde{a}_{\tau,ijk\ell}\hspace{50 \mu}$, are also decision variables in the new SOS program.

The time evolutions of $\overline{H}_d$, the inputs, and the corresponding state trajectories for four different initial conditions are shown in Figure \ref{fig:H_u_states}. For the chosen set of initial conditions, the results clearly show the asymptotic stability of the desired set-point. Moreover, the Hamiltonian asymptotically converges to zero. 

\begin{figure}[tbp]
    \centering
    \begin{minipage}[b]{0.48\columnwidth}
        \centering
        \includegraphics[width=\textwidth]{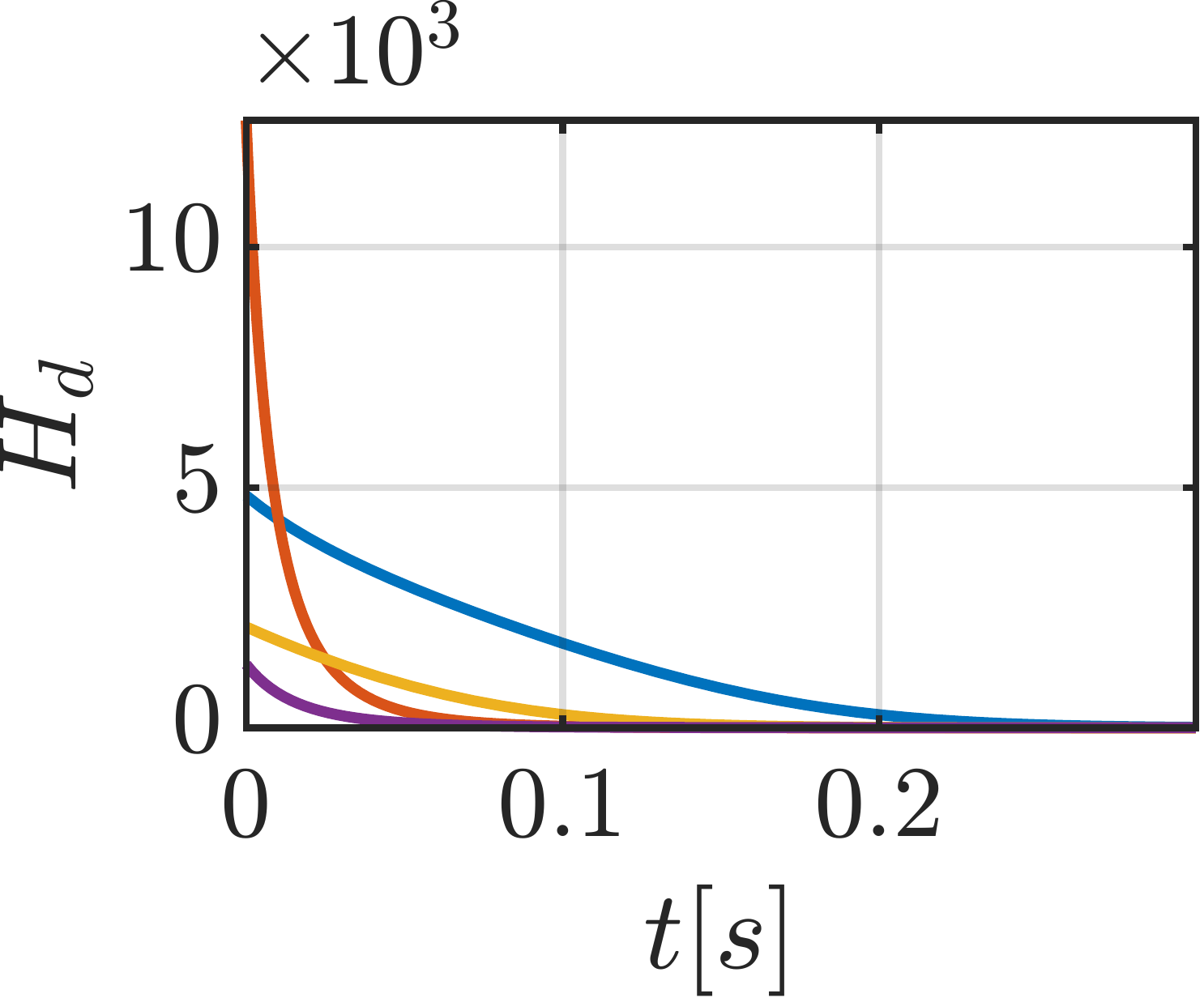}
    \end{minipage}
    \hfill
    \begin{minipage}[b]{0.48\columnwidth}
        \centering
        \includegraphics[width=\textwidth]{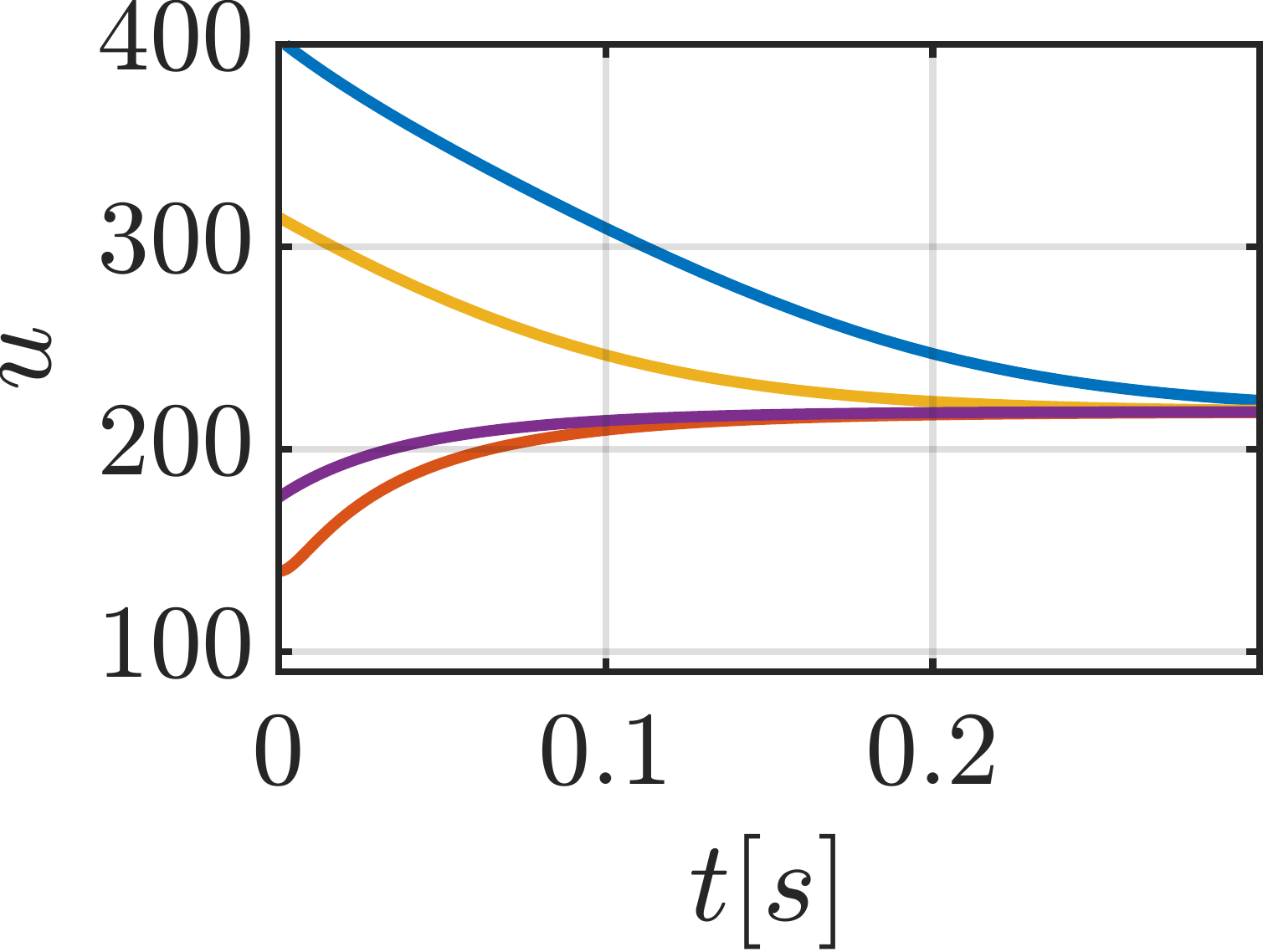}
    \end{minipage}
    \begin{minipage}[b]{0.48\columnwidth}
        \centering
        \includegraphics[width=\textwidth]{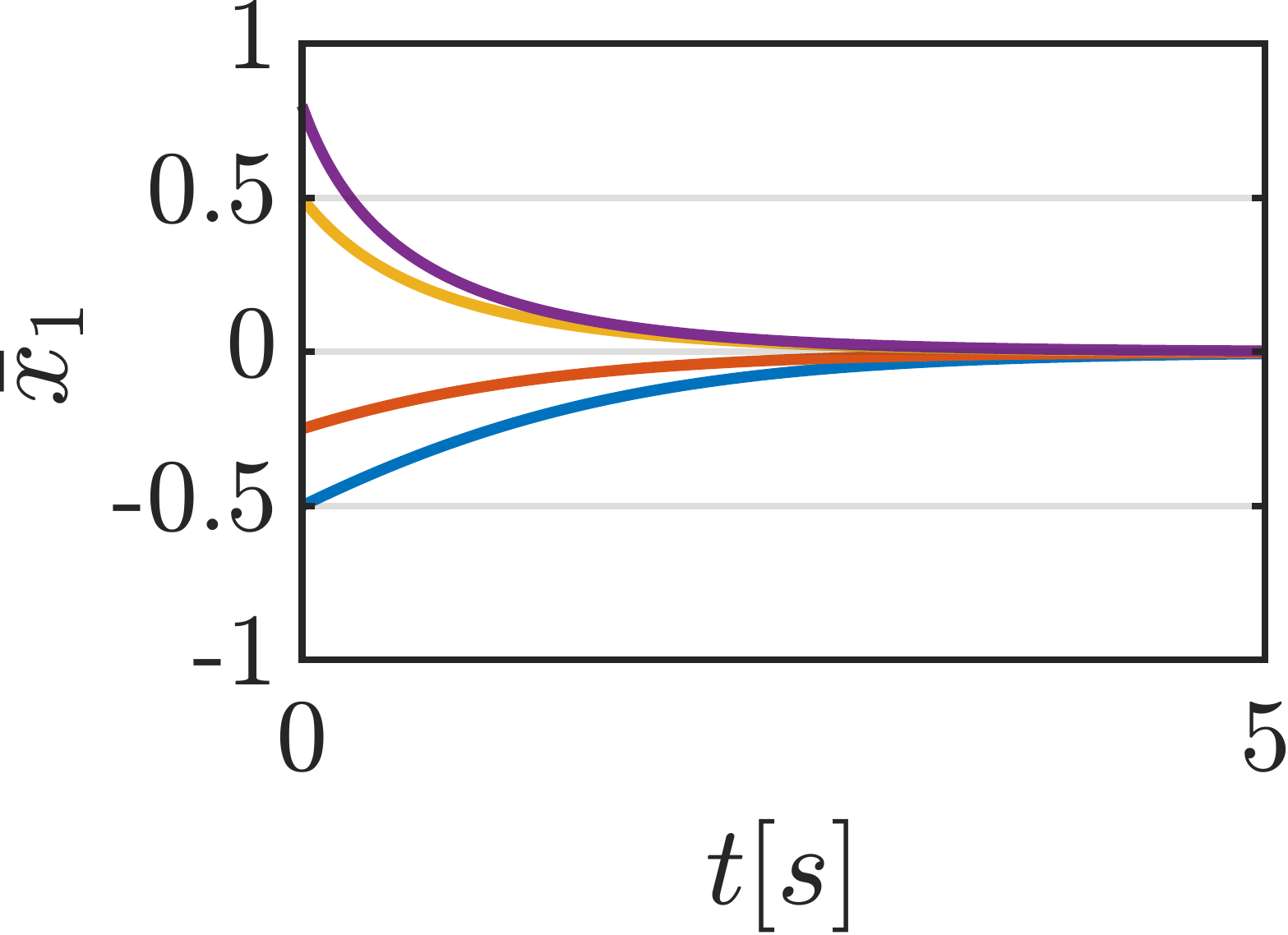}
    \end{minipage}
    \hfill
    \begin{minipage}[b]{0.48\columnwidth}
        \centering
        \includegraphics[width=\textwidth]{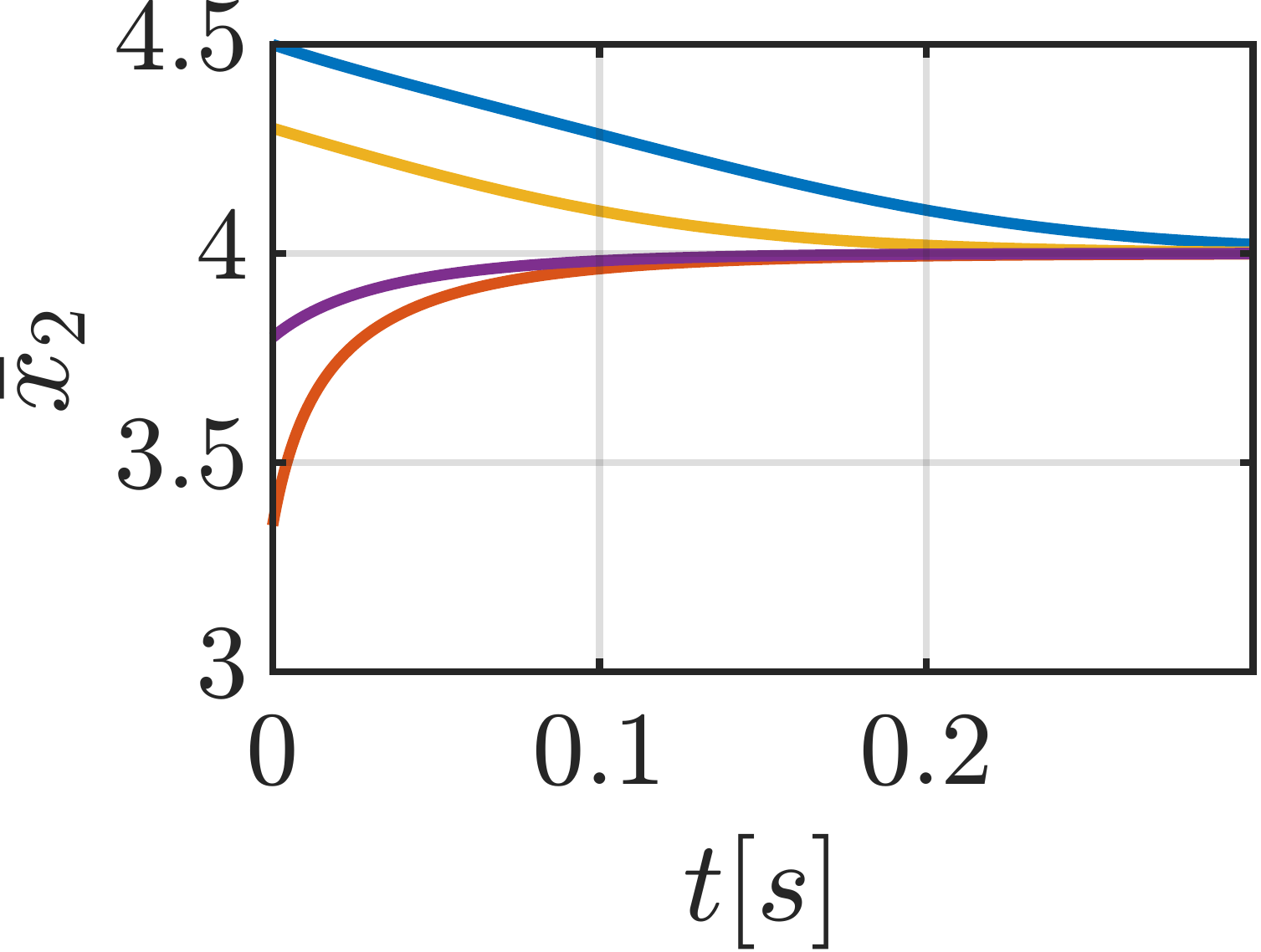}
    \end{minipage}
    \begin{minipage}[b]{0.48\columnwidth}
        \centering
        \includegraphics[width=\textwidth]{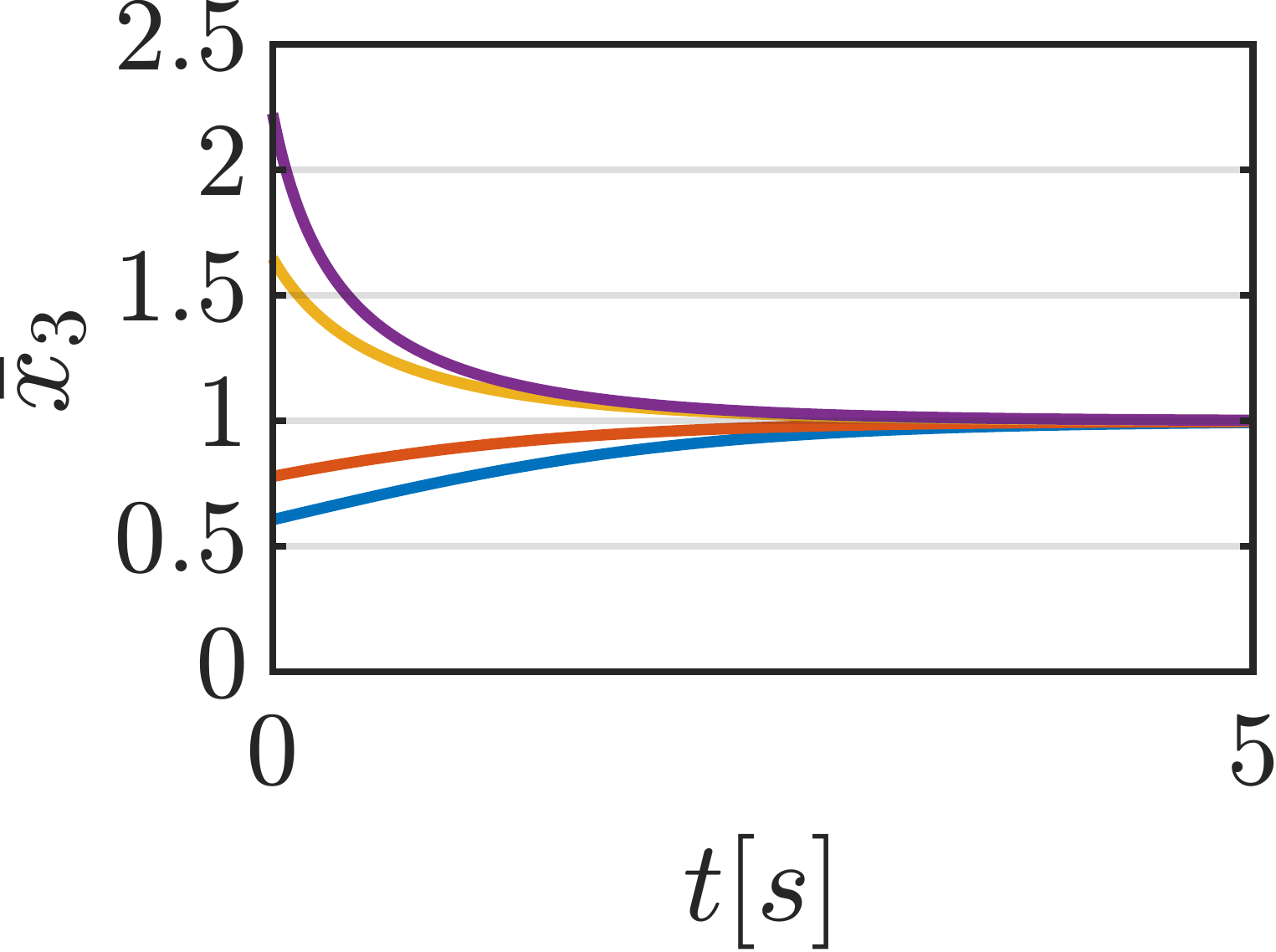}
    \end{minipage}
    \hfill
    \begin{minipage}[b]{0.48\columnwidth}
        \centering
        \includegraphics[width=\textwidth]{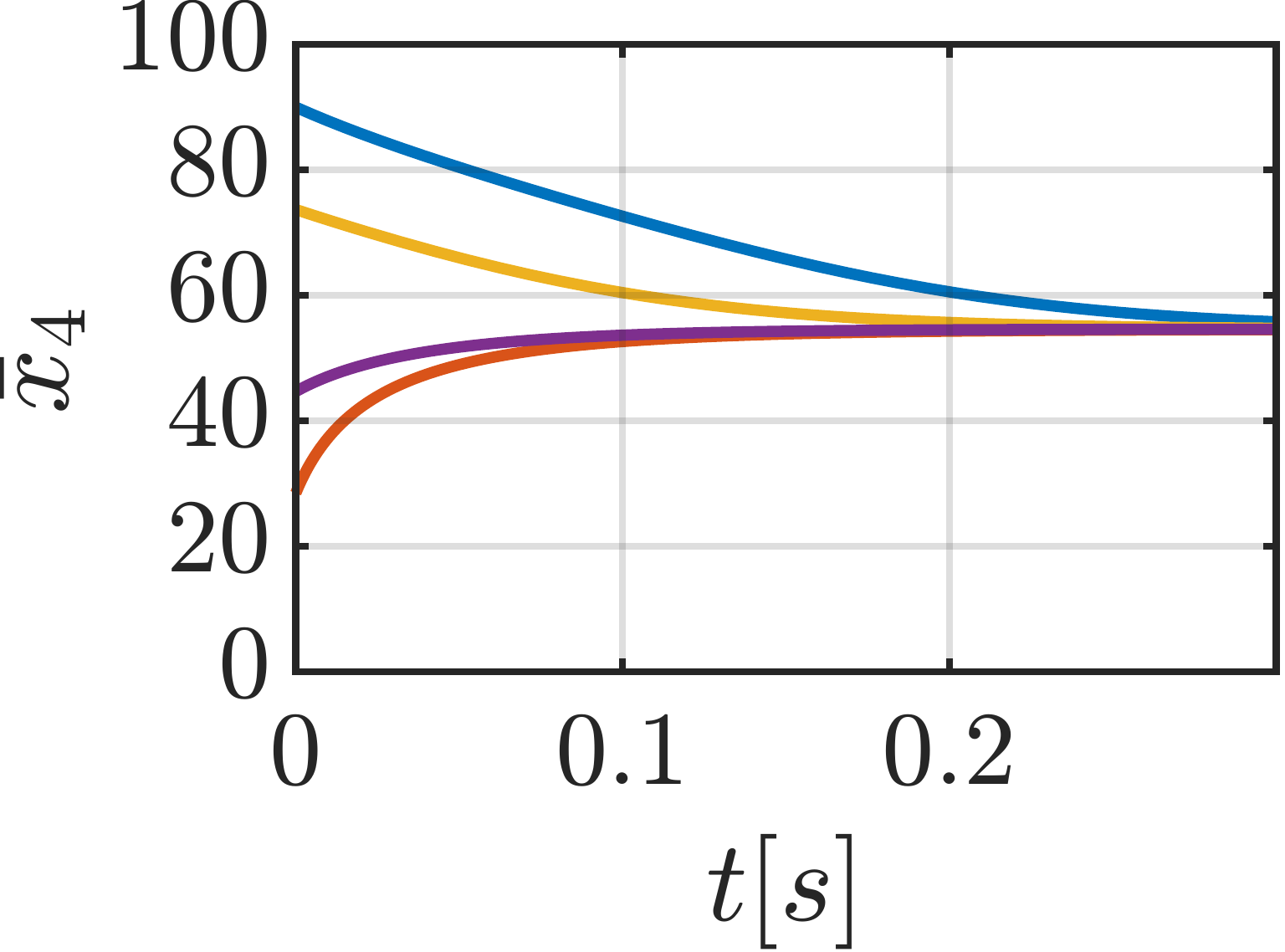}
    \end{minipage}

    \caption{Time evolution of the Hamiltonian $H_d$ (top left), the input $u$ (top right), and the states.}
    \label{fig:H_u_states}
\end{figure}
%--------------------
\section{Conclusion and Outlook}
\label{sec:conclusion} 
This paper proposed an approach that captures the behavior of non-polynomial port-Hamiltonian systems within a higher-dimensional polynomial system representation without sacrificing their core structure. Our approach guarantees the preservation of the interconnection geometry, the dissipated energy, and passivity on the range of the immersion. Moreover, for consistently initialized initial conditions, this subset is a forward invariant manifold under the obtained dynamics. Additionally, we briefly discussed a SOS control design that leverages the polynomial structure obtained via the immersion.

Future work will investigate further exploitation of polynomial immersions for control design and data-driven modeling. 
%--------------------
	\bibliographystyle{elsarticle-num-names}
	\bibliography{Polynomial_immersion_references}
\end{document}